\documentclass[11pt]{article}

\usepackage[utf8]{inputenc}
\usepackage[T1]{fontenc}
\usepackage{amsmath}
\usepackage{amssymb}
\usepackage{amsthm}
\usepackage{mathtools}
\usepackage{graphicx}
\usepackage{booktabs}
\usepackage{algorithm}
\usepackage{algpseudocode}
\usepackage{xcolor}
\usepackage{hyperref}
\usepackage{cleveref}
\usepackage{rotating}
\usepackage{makecell}
\usepackage{booktabs}
\usepackage{xspace}
\usepackage{subcaption}
\usepackage{pdflscape}
\usepackage{relsize}
\usepackage{adjustbox} 
\usepackage{tikz}
\usetikzlibrary{arrows.meta, positioning, calc}

\DeclareMathOperator{\supp}{supp}
\DeclareMathOperator{\set}{set}

\newtheorem{theorem}{Theorem}
\newtheorem{lemma}{Lemma}
\newtheorem{proposition}{Proposition}
\theoremstyle{definition}
\newtheorem{definition}{Definition}

\renewcommand{\triangledown}{\begin{sideways}\begin{sideways} \ensuremath{\triangle} \end{sideways}\end{sideways}}

\newcommand{\triangledownscript}{\begin{sideways}\begin{sideways} $\scriptstyle{\triangle}$ \end{sideways}\end{sideways}}

\newcommand{\sactivity}{{\ensuremath{\blacktriangleright}}\xspace}
\newcommand{\eactivity}{\ensuremath{\mathsmaller{\blacksquare}}\xspace}

\title{Monotonicity-Guided Bottom-Up Petri Net Discovery: The \textsc{SPECpp} Framework}
\author{Leah Tacke genannt Unterberg \and Lisa L. Mannel \and Wil M. P. van der Aalst\footnote{Corresponding author \url{wvdaalst@pads.rwth-aachen.de}}}
\date{\small Process And Data Science (PADS) \\ RWTH Aachen University \\ D-52074 Aachen, Germany}

\begin{document}

\maketitle

\begin{abstract}
Process discovery is one of the central challenges in process mining. 
Petri nets are particularly attractive because simple local constructs can express complex behavior, including concurrency. 
While their global behavior may be difficult to analyze, individual places can be efficiently characterized using monotonic properties, enabling bottom-up discovery. 
Unlike top-down approaches such as the Inductive Miner, which rely on predefined constructs for sequences, choices, loops, and concurrency, our approach allows such structures to emerge organically and can exploit the full expressive power of Petri nets, including free-choice constructs and long-term dependencies. 
The main challenge is the exponential number of candidate places and their combinations. 
We present the \textsc{SPECpp} framework which implements strategies to obtain high-quality models under time and resource constraints. 
\textsc{SPECpp} supports rapid experimentation and is used to evaluate these strategies using both synthetic and real-life event data.
\end{abstract}

%

\noindent\textbf{Keywords:} Bottom-up process discovery $\cdot$ Petri nets $\cdot$ Framework

\section{Introduction}
\label{sec:introduction}

Ever since the inception of process mining, process discovery has been the most researched task within the discipline \cite{vanderAalst2016DataScience,PMhandbook-SS22}. The automated, data-based inference of a process model that accurately portrays reality and lends itself to further analysis is immensely appealing. Understandably, process discovery is a key capability in the portfolio of process mining vendors like Celonis, SAP Signavio, Apromore, etc. Their increasing commercial success drives industry adoption, thereby opening up the field to new application domains and unseen levels of volume and variety, which again reinforces the need for further development of the core technologies.

While simpler representations such as directly-follows graphs remain popular \cite{centeris-keynote2019}, there is a clear need for executable, precise, and human-readable process models able to express concurrency. 
Such models can serve as a basis for simulation, explainable prediction, and ultimately automated process improvement.

From early work such as the Alpha Miner~\cite{vanderAalst2004WorkflowMining} to the ever popular Inductive Miner~\cite{Leemans2013DiscoveringBlockStructured}, all discovery algorithms directly and indirectly introduce representational and modeling biases in part by virtue of their structure. 
Some algorithms can be clearly categorized as functioning in a top-down or bottom-up manner \cite{vanderAalst2016DataScience,Ch2-PMhandbook-SS22}. Informally, they either try to detect a global structure and subdivide the problem, or agglomerate atomic behavioral patterns. While there is no agreed-upon definition in the relevant literature, we use the following intuition. The inductive miner is an example of the former class, as it extracts the ``outermost'' behavior first, and recursively refines the inner behavior of the resulting independent sub models. This type of top-down specification provides ample structure which is useful for giving guarantees on the model level. However, it is also the reason for a major weakness: it is assumed that models are block-structured. Therefore, this divide-and-conquer style of independent sub problems introduces the limitation of not being able to discover dependencies spanning multiple sub models (blocks). We discuss related work further in \Cref{sec:related-work}.

In contrast to that, bottom-up approaches do not have to concern themselves with prior overall model structure specifications. The structure organically emerges from the agglomeration of minimal constructs. Put in another way: instead of attempting to detect structures in the input and translating it to the basic building blocks of the modeling formalism, we start with these basic building blocks and see which fit the input. For Petri nets, these minimal constructs are places. While their combination is intractable to analyze (Petri net reachability is EXPSPACE-hard~\cite{Lipton1976ReachabilityProblem}) and make general guarantees about without drastically limiting the expressiveness (e.g., block-structured models), individual places make for great candidates for efficiently reasoning on. Constraints on the whole model, like minimum fitness, are reflected on its individual places. By putting constraints on the individual places that eventually form the Petri net, multiple of the well-known model quality dimensions fitness, precision, simplicity and generalization~\cite{Buijs2012Role} can be balanced starting from the low-level.

We propose a conceptual framework named \textsc{SPECpp} that structures this kind of discovery into iterative candidate Proposal, Evaluation and Composition (PEC-cycle) followed by post-processing (pp) for cleanup. The ``S'' of the acronym comes from the implementation which we detail in \Cref{sec:implementation}. \Cref{fig:framework-overview} provides a high-level overview. Essentially, we directly, and even better, indirectly, consider all possible places (like ILP-based approaches) and collect those passing an evaluation step in an intermediate result. This intermediate result can inform new place proposal, e.g., to guide the candidate search. The main strategy to combat intractability is exploiting monotonicity properties of places and place sets. A second, complementary strategy is to employ greedy search to further restrict the number of candidates considered.
In our current base framework instantiation based on eST-Miner variants~\cite{Mannel2019FindingComplex,Mannel2019FindingUniwired,Mannel2022DiscoveringProcessModels,Mannel2020RemovingImplicit,Mannel2020ImprovingState} by Mannel et al., we inherit the limitation of only supporting uniquely labeled transitions (so no silent/tau transitions either). However, we can use pre-processing approaches like~\cite{Li2007ExtendingAlpha} to address this.

\begin{figure}
  \centering
  \includegraphics[width=\textwidth]{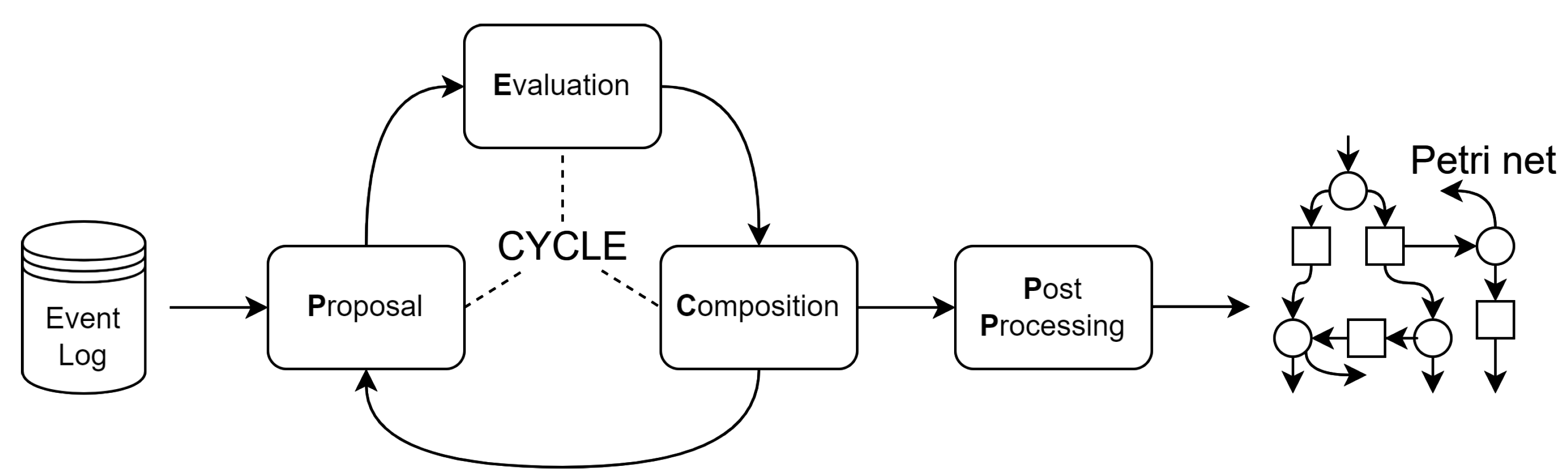}
  \caption{A high-level overview of the framework flow. An event log is taken as input to inform a \emph{proposal}, \emph{evaluation} and \emph{composition} loop of candidate places. The final resulting Petri net is constructed after a \emph{post-processing} step.}
  \label{fig:framework-overview}
\end{figure}

To motivate the necessity of bottom-up approaches like this, we just need to turn our head to process control flows with complex dependencies that are not easily separated. For example, so-called long-term dependencies and short loops. We adapt a minimal example from Mannel et al.~\cite{Mannel2019FindingComplex}. Imagine a delivery process where in the beginning either a regular or a VIP-customer's order comes in (denoted by activities \texttt{reg.\ order} ($o$) and \texttt{VIP order} ($vo$)). Then, a varying number of deliveries (\texttt{delivery} ($d$)) are made for that order. Finally, the customer is invoiced the regular price or the VIP special price (\texttt{reg.\ invoice} ($i$) and \texttt{VIP invoice} ($vi$)). Of course, the invoice type is in accordance with the earlier order type. A sample event log of records of this process could be $[\langle \sactivity, o, d, d, i, \eactivity \rangle, \langle \sactivity, vo, d, d, d, vi, \eactivity \rangle]$, with $\sactivity$ and $\eactivity$ denoting the start and end of a process instance respectively. \Cref{fig:delivery-example} shows the results of various established discovery algorithms on this input. 
Among existing approaches, only the eST-Miner correctly captures both the long-term dependency between order and invoice type and the delivery self-loop.

\begin{figure}[htbp]
    \centering
\begin{adjustbox}{width=\textwidth, totalheight=0.8\textheight, keepaspectratio}
        \begin{tabular}{c} 
            
            \begin{subfigure}{\textwidth}
                \centering
                \includegraphics[width=\linewidth]{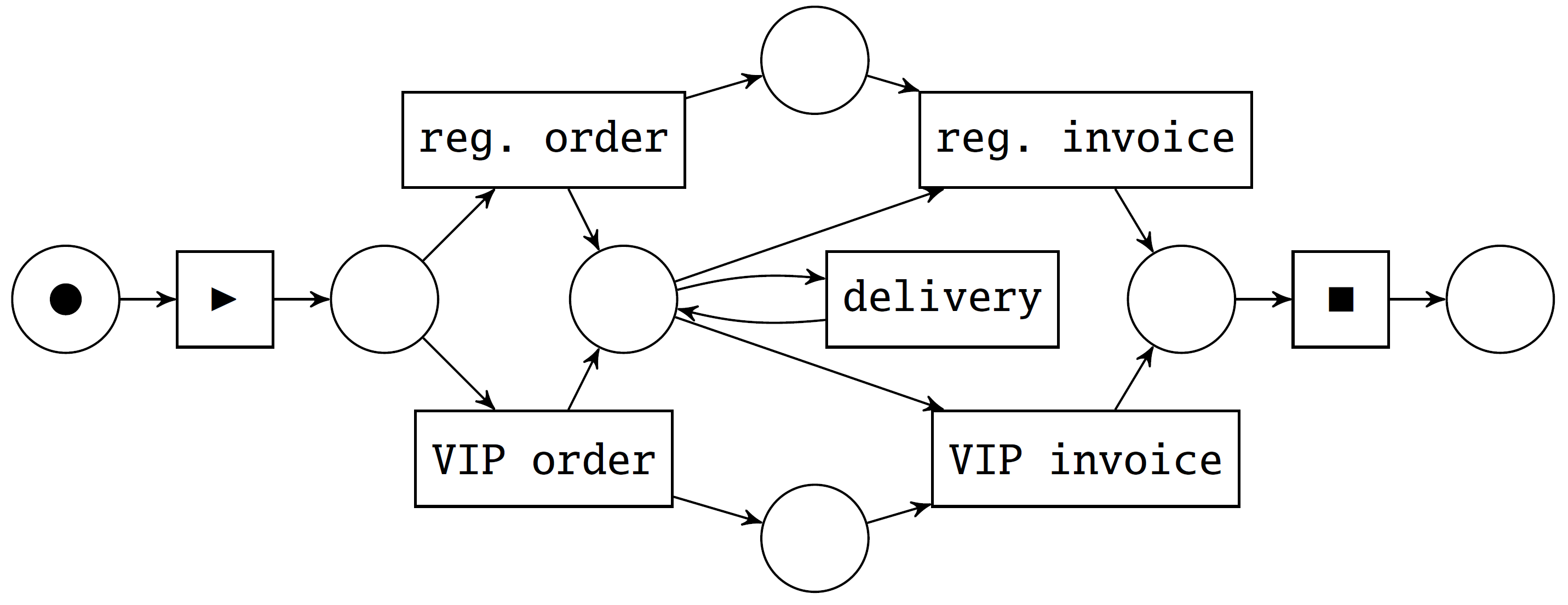}
                \caption{Base eST-Miner}
                \label{fig:delivery-alg1}
            \end{subfigure} \\ [1em]

            \begin{subfigure}{\textwidth}
                \centering
                \includegraphics[width=\linewidth]{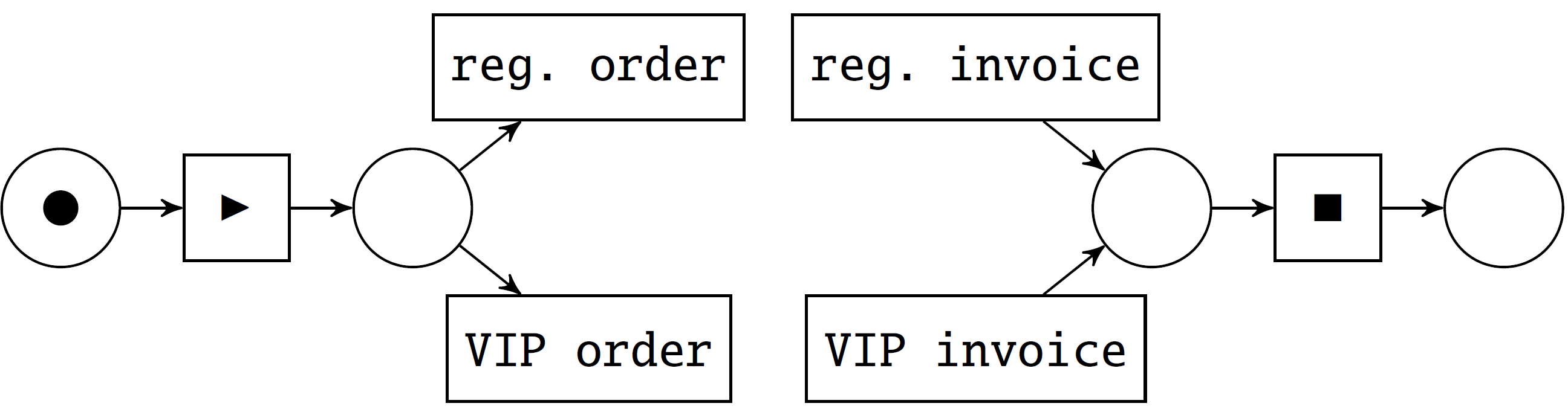}
                \caption{Alpha Miner}
                \label{fig:delivery-alg2}
            \end{subfigure} \\ [1em]

            \begin{subfigure}{\textwidth}
                \centering
                \includegraphics[width=\linewidth]{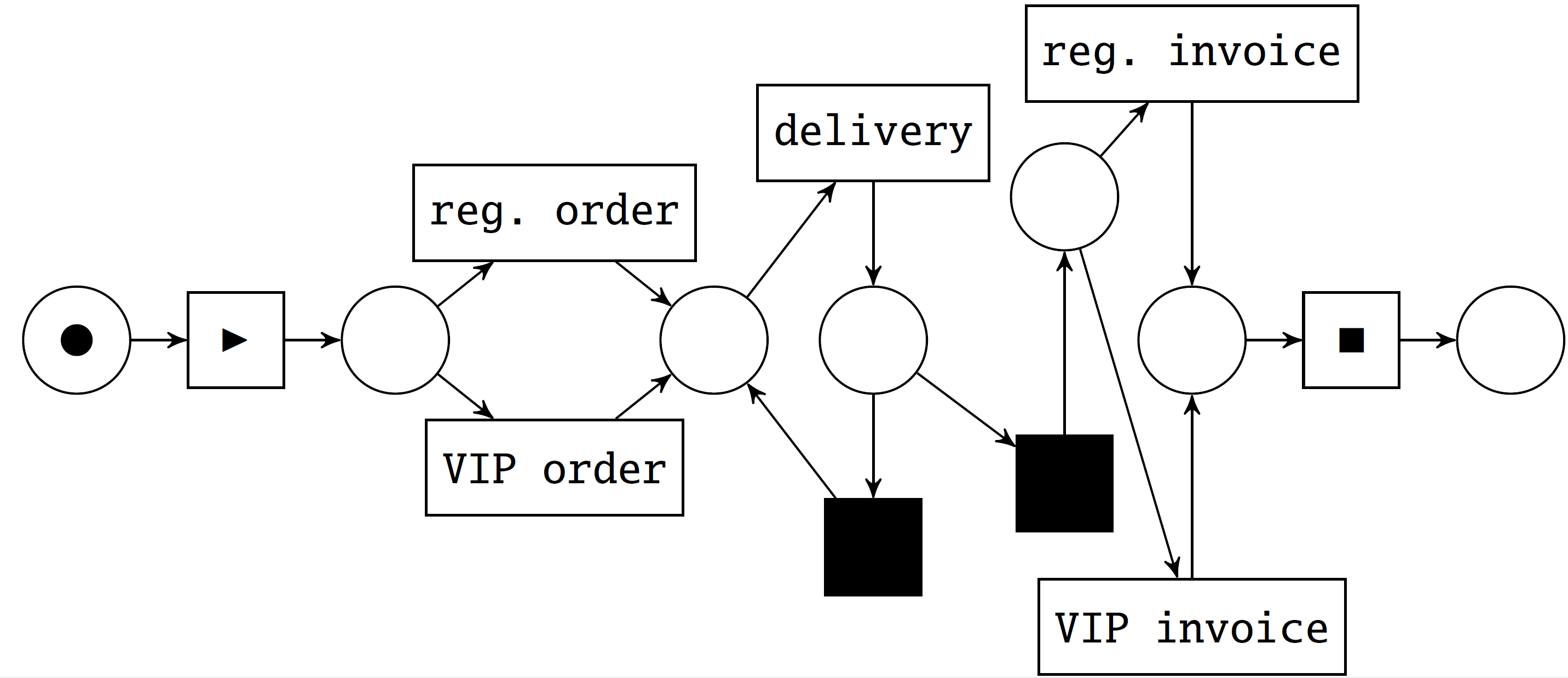}
                \caption{Inductive Miner}
                \label{fig:delivery-alg3}
            \end{subfigure} \\ [1em]

            \begin{subfigure}{\textwidth}
                \centering
                \includegraphics[width=\linewidth]{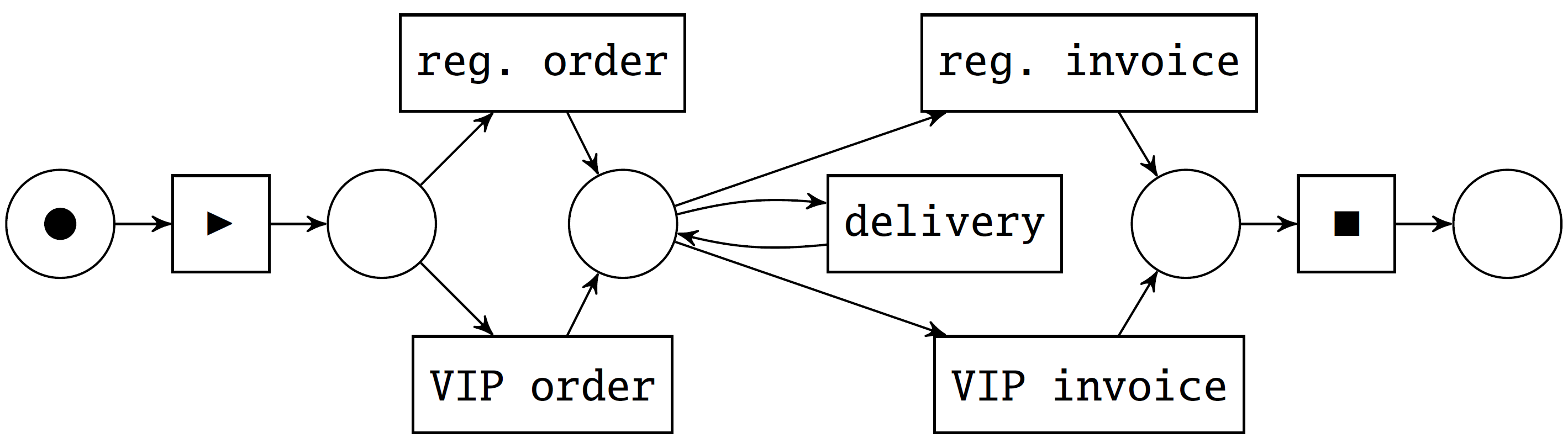}
                \caption{ILP Miner}
                \label{fig:delivery-alg4}
            \end{subfigure}

        \end{tabular}
    \end{adjustbox}

    \caption{Petri nets discovered by various algorithms on the input log \linebreak
    $[\langle \sactivity, o, d, d, i, \eactivity \rangle, 
    \langle \sactivity, vo, d, d, d, vi, \eactivity \rangle]$. 
    The long-term (non-directly following) dependency between VIP or regular 
    types, particularly with the looping activity \texttt{delivery}, provides 
    a challenge for established discovery algorithms.}
    \label{fig:delivery-example}
\end{figure}

It is important to emphasize that there is never ``one correct'' model. Model quality is informed by its purpose. The conceptual framework we present in this paper is highly flexible and configurable to be adaptable to the user's needs to balance the aforementioned quality dimensions.
Bottom-up discovery requires two key decisions: which candidate places to consider and how to combine accepted places into a model. We call these candidate proposal and candidate composition, respectively. Candidate composition is informed by candidate evaluation. Together, proposal, evaluation, and composition form the PEC cycle, which is followed by a post-processing stage.
Next to the conceptual framework we present an extensive implementation (which we introduce in \Cref{sec:implementation}) serving as a software framework for fast prototyping of varied metrics and constraint types that guide the discovery phase.

The rest of this paper is structured as follows. We first provide a formal background in \Cref{sec:background}, before we introduce our conceptual framework in \Cref{sec:conceptual-framework}. Then, we go over the related implementation artifacts that we published in \Cref{sec:implementation}. We then include a brief evaluation in \Cref{sec:experiments} to validate our approach. After positioning this framework among related work in \Cref{sec:related-work}, we conclude the paper with some discussion and outlook on future work in \Cref{sec:conclusions}.

\section{Background}
\label{sec:background}

We use sets $X = \{x_1, x_2, \ldots\}$, multisets $M = [x_1^{f_1}, x_2^{f_2}, \ldots]$ and finite ordered sequences $\sigma = \langle s_1, s_2, \ldots, s_n \rangle \in S^*$ with the usual semantics. Elements in the support set of a multiset $x_i \in \supp(M)$ have positive frequency $f_i \in \mathbb{N}$ and $\sigma(i) = s_i$ with $i \in \{1, \ldots, n\}$ refers to the $i$-th element of a sequence. The elements of a sequence are $\set(\sigma) = \{\sigma(1), \ldots, \sigma(n)\}$. $|\cdot|$ denotes the cardinality, i.e., number of elements, in a set, multiset or sequence. We use $\underline{k} = \{1, \ldots, k\}$ as a shorthand for integer index sets (with $\underline{0} = \emptyset$). The set of possible finite sequences over an alphabet $\Sigma$ is $\Sigma^* = \bigcup_{n \in \mathbb{N}_0} \Sigma^n$, where $\Sigma^n$ denotes all sequences of length $n \in \mathbb{N}_0$. The powerset over $X$ is denoted by $\mathcal{P}(X)$ whereas $\mathcal{B}(X)$ refers to all multisets over $X$.

Further, we use permutations of sets $S$, $\pi \in \mathrm{Perm}(S)$ with $\pi : S \to \underline{|S|}$ such that $\forall i, j \in \underline{n} : i \neq j \Leftrightarrow \pi(i) \neq \pi(j)$, i.e., $\pi$ is bijective. A permutation uniquely maps a
set to its ordered sequence by specifying element indices. A total ordering $\leq$ on $X$ (a reflexive, antisymmetric and transitive relation) which further relates all elements, i.e., $\forall x, y \in X : x \leq y \vee y \leq x$, induces a permutation $\pi_{\leq}$ with $\forall x, y \in X : \pi_{\leq}(x) \leq \pi_{\leq}(y) \Leftrightarrow x \leq y$.

Additionally, we use predicate functions over sets $X$ of the form $f : X \to \{\mathrm{True}, \mathrm{False}\}$. We usually define them with logical expressions $\Rightarrow, \vee, \wedge, \neg$ over the elements of $X$. They are used as filters on subsets $S \subseteq X$ with the following notation: $S \upharpoonright_f = \{s \in S \mid f(s) = \mathrm{True}\}$. Restrictions on functions restrict the domain, i.e., a function $g : X \to Y$ is restricted by $f$ to $\mathrm{dom}(g \upharpoonright_f) = X \upharpoonright_f$.

In process mining, events are typically differentiated from their associated activity, however, we only care about event logs as sequences of activities in our setting. There are many opportunities for incorporating further event attributes into discovery or enriching the resulting models~\cite{Mannhardt2015MultiPerspective}. Though the implementation artifacts of this work support such extensions, particularly for heuristics, we leave this as future work.

\begin{definition}[Activities, Behaviors]
\label{def:activities-behaviors}
$\mathbb{A}$ is the universe of activities (actions, tasks, etc.). $\{\sactivity, \eactivity\} \subseteq \mathbb{A}$ are unique start and end activities. A \emph{behavior} is a sequence $\sigma = \langle a_1, \ldots, a_n \rangle \in \mathbb{A}^*$ of activities that starts with $\sactivity$ and ends with $\eactivity$, i.e., $a_1 = \sactivity \wedge a_n = \eactivity$ and $\forall 1 < i < n : a_i \notin \{\sactivity, \eactivity\}$. The universe of possible behaviors is then given by $\mathbb{B} = \{\sigma \in \mathbb{A}^* \mid \sigma \text{ is a behavior}\}$.

Note that, while sequences can be empty, a behavior has at least length two. The fixed start/end activities pose no threat to generality, as sequences are trivially transformed to behaviors. We require them here to not have to deal separately with initially and finally occurring activities. It also simplifies the construction of the resulting Petri net.
\end{definition}

\begin{definition}[Event Log]
\label{def:event-log}
An event log $L$ is represented as a multiset of behaviors, i.e., $L \in \mathcal{B}(\mathbb{B})$. In this context, behaviors are also referred to as \emph{traces}.
\end{definition}

For the notation of our resulting process model, we use a Petri net based formalism. Specifically, uniquely labeled Petri nets represented by sets of places. Transitions and places are the fundamental building blocks of Petri nets. We uniquely identify transitions with activities in this work, so we may use the terms interchangeably.

\begin{definition}[Places]
\label{def:places}
$\mathbb{P} = \mathcal{P}(\mathbb{A}) \times \mathcal{P}(\mathbb{A})$ is the universe of all possible places. For $p = (I, O) \in \mathbb{P}$, we call ${\bullet}p = I$ its preset and $p{\bullet} = O$ its postset. Transitions in both, i.e., the intersection ${\bullet}p \cap p{\bullet}$ are called \emph{self-loops}. Places $\mathbb{P}^! = (\mathcal{P}(\mathbb{A}) \setminus \{\emptyset\}) \times (\mathcal{P}(\mathbb{A}) \setminus \{\emptyset\})$ with non-empty pre- and postset are called \emph{through places}.
\end{definition}

An important structural relation on places which we will be using are preset and postset expansions.

\begin{definition}[Pre- \& Postset Expansions]
\label{def:expansions}
Let $p, p' \in \mathbb{P}$ be places. $p'$ is called a preset expansion of $p$ if and only if ${\bullet}p \subset {\bullet}p'$ and $p{\bullet} = p'{\bullet}$ hold. It is \emph{direct} if ${\bullet}p$ and ${\bullet}p'$ differ by exactly one element. Postset expansions are defined symmetrically.

As defined here, we only consider pure expansions, i.e., those where only either the pre- or postset changes. For an example, consider $(\{a\}, \{b, c\})$ which is a direct postset expansion of $(\{a\}, \{b\})$ and an indirect one of $(\{a\}, \emptyset)$.
\end{definition}

Usually, the semantics of Petri net places are defined via markings, that is, counts of tokens on a place. When regarding a place and its adjacent transitions in isolation, the transitions in its postset are enabled whenever its number of tokens is positive. Of course, from the perspective of a transition, all of the places it is in the postset in need to have a positive number of tokens for it to be enabled. Enabled postset transitions can be executed to consume one token. The transitions exclusively in its preset (so non self-loops) are not constrained (by this particular place) and can fire at any time to produce a token. In our setting, this is equivalent to the property that for any prefix of a behavior, the number of activities in the preset up to the penultimate element of the prefix is greater or equal to that of those in the postset for the entire prefix. Additionally, they are equal for the entire behavior which implies equal token counts at the start and end of a trace, given the former condition holds. For our place fitness, we consider places as initially and finally unmarked. In our Petri net conversion, we add a marked start place and designated sink place, so this is not a restriction. This class is close to \emph{workflow nets} without the requirement that all transitions lie on a path between start and end.

\begin{definition}[Fitness]
\label{def:fitness}
Let $p \in \mathbb{P}$ be a place and $\sigma = \langle a_1, \ldots, a_n \rangle \in \mathbb{B}$ a behavior. The behavior $\sigma$ is \emph{fitting} on $p$, i.e., $\square_\sigma(p)$ if and only if $\text{\textsc{non-neg}}(\sigma, p) \wedge \text{\textsc{balanced}}(\sigma, p)$ holds. Where these properties are defined as follows.
\begin{align*}
\text{\textsc{non-neg}}(\sigma, p) &:\Leftrightarrow \forall k \in \underline{n} : \text{\textsc{non-neg}}_k(\sigma, p) \\
\text{\textsc{non-neg}}_k(\sigma, p) &:\Leftrightarrow |\{i \in \underline{k-1} \mid \sigma(i) \in {\bullet}p\}| \geq |\{i \in \underline{k} \mid \sigma(i) \in p{\bullet}\}| \\
\text{\textsc{balanced}}(\sigma, p) &:\Leftrightarrow |\{i \in \underline{n} \mid \sigma(i) \in {\bullet}p\}| = |\{i \in \underline{n} \mid \sigma(i) \in p{\bullet}\}|
\end{align*}
The set of fitting behaviors on $p$ is then given by $\mathit{fit}(p) = \{\sigma \in \mathbb{B} \mid \square_\sigma(p)\}$.
\end{definition}

This notion of fitness intuitively corresponds to the classic token-based replay on Petri nets~\cite{Rozinat2008ConformanceChecking}, e.g., during replay of behavior $\langle \sactivity, a, b, a, b, \eactivity \rangle$ on place $(\{a\}, \{b\})$, there are no missing tokens (\textsc{non-neg}), and at the end, there is no remaining token (\textsc{balanced}).

While one place only constrains its adjacent activities, making it efficient to analyze and reason about, the behavior of a set of places is simply the intersection of its constituents.

\begin{definition}[Sets of Places]
\label{def:sets-of-places}
Let $P \subseteq \mathbb{P}$ be a set of places and $\sigma \in \mathbb{B}$ a behavior. $\sigma$ is fitting on $P$ if and only if $\square_\sigma(P) :\Leftrightarrow \forall p \in P : \square_\sigma(p)$. The set of fitting behaviors on $P$ is $\mathit{fit}(P) = \{\sigma \in \mathbb{B} \mid \square_\sigma(P)\} = \bigcap_{p \in P} \mathit{fit}(p)$.

\end{definition}
This definition implies that any trace is fitting on an empty set of places, which intuitively makes sense as places can be regarded as constraints.

A set of places can be trivially translated to an equivalent Petri net with open-world semantics. That is, contrary to the typical closed-world semantics, an activity not contained in the model is unconstrained and can be executed at any time.

\begin{definition}[Petri net Conversion]
\label{def:petri-net-conversion}
Let $P \subseteq \mathbb{P}$ be a set of places. The set of occurring activities is $A = \bigcup_{p \in P} {\bullet}p \cup p{\bullet}$. Then $N = (P', A, F, M_{\mathrm{init}}, M_{\mathrm{final}})$ with places $P' = P \cup \{p_{\mathrm{init}}, p_{\mathrm{final}}\}$ using $p_{\mathrm{init}} = (\emptyset, \sactivity)$ and $p_{\mathrm{final}} = (\eactivity, \emptyset)$ is a behaviorally equivalent marked Petri net under open-world semantics. The arcs are uniquely specified by the places $F = \{(i, p), (p, o) \mid p \in P', i \in {\bullet}p, o \in p{\bullet}\} \subseteq (P' \times A) \cup (A \times P')$. The initial and final marking are simply $M_{\mathrm{init}} = [p_{\mathrm{init}}]$ and $M_{\mathrm{final}} = [p_{\mathrm{final}}]$.
\end{definition}

Note that, as $p_{\mathrm{init}}$ and $p_{\mathrm{final}}$ have no fitting behavior according to $\square$, they are artificially added for marking-based semantics. Given a fixed set of activities, e.g., those considered possible in a process, the Petri net could also be extended with a concurrent ``flower model gadget'' to preserve the same fitting behavior under the closed-world semantics. In practice, this is not really a limitation as the given event log used for discovery is assumed to be complete in that regard.

\section{Conceptual Framework}
\label{sec:conceptual-framework}

We present a framework for bottom-up discovery of sets of places. There are two natural hooks for strategies: which places to consider and which combinations of places to consider. We term the strategy of which places to consider candidate \emph{proposal} and which combinations to consider candidate \emph{composition}. Composition strategies rely on a preceding candidate \emph{evaluation} for their decision-making. Together, these three steps form the \emph{PEC-cycle}. It is repeated until the proposal strategy is exhausted (no remaining candidates) but may also be stopped early, e.g., by specifying a time limit. After this loop terminates, we append a one-shot post-processing pipeline. This is also where we convert from sets of places, which we use formally and internally, to Petri nets. In this work, we focus mainly on the earlier stages, however, there is a lot of research and development potential in the latter stages that we want to support. Refer to \Cref{fig:framework-overview} from the introduction for a high-level schematic overview.

For this section, we fix a finite set of considered activities $A \subseteq \mathbb{A}$ and thus a set of possible places $P = \{p \in \mathbb{P} \mid {\bullet}p \cup p{\bullet} \subseteq A\}$. In an application setting, this would typically be the set of activities occurring in the event log used as input. Then, the number of possible candidates $|P|$ is in $\mathcal{O}(2^{|A|})$ and the number of possible combinations of those candidates $|\mathcal{P}(P)|$ is in $\mathcal{O}(2^{|P|}) = \mathcal{O}(4^{|A|})$. This makes a brute-force approach wholly intractable. By exhaustive enumeration, though infeasible even for small inputs, we could trivially guarantee finding the best possible solution according to any predefined quality metric.

Luckily, there exist powerful monotonicity properties on places and place sets which enable us to drastically prune the search space. An evaluation of one candidate place can allow us to indirectly reason about a whole set of monotonically related places. Similarly, the conjunction of place properties on a place set may allow us to indirectly consider all its supersets. Additionally, we employ greedy decision-making and place set invariants (e.g., implicitness freedom).

We first introduce our conceptual proposal strategy and then move over to evaluation together with composition and finally post-processing.

\subsection{Proposal}
\label{sec:proposal}

The primary method used in our framework for limiting the number of proposed candidates is the usage of a \emph{candidate tree} in which the subtree relation is aligned to monotonicity properties on places. Further, we require a compact representation of that tree as it still has exponential size. We achieve this by using a locally computable child generation function instead of pre-computing the complete tree. The local computability allows us to incrementally unfold the tree while ``forgetting'' fully traversed nodes.

\begin{definition}[Efficient Candidate Tree]
\label{def:candidate-tree}
An efficient candidate tree $T = (g, r)$ is defined by a child generation logic $g : P \to \mathcal{P}(P)$ and a root node $r \in P$. The following properties of $g$ are required so that $T$ actually unfolds into a tree. Let $g^{-1}(p) = \{p' \in P \mid p \in g(p')\}$ for $p \in P$ be the inverse of $g$. The root has an empty set of parents, i.e., $g^{-1}(r) = \emptyset$.
\begin{align}
\forall p \in P &: p \notin g(p) \label{eq:tree1} \\
\forall p \in P \setminus \{r\} &: |g^{-1}(p)| = 1 \label{eq:tree2}
\end{align}

Further, and most importantly, $g$ has to be acyclic. For any sequence of at least two places $\langle p_1, \ldots, p_n \rangle \in P^*$, $n \geq 2$ which represents a path through $T$, i.e., $\forall 1 \leq i < n : p_{i+1} \in g(p_i)$, we require the start and end to be unique, i.e., $p_1 \neq p_n$.

\end{definition}

The acyclicity condition already implies condition~\eqref{eq:tree1}. Equation~\eqref{eq:tree2} states that each node, except the root, has to have a unique parent. Together with the finiteness of $P$, it follows that there must exist some leaf nodes without children. The places contained in a subtree rooted at a place $p \in P$ under generation logic $g$ are then well-defined by the following recursion.
\[
\mathit{cand}_g(p) = \{p\} \cup \bigcup_{p' \in g(p)} \mathit{cand}_g(p')
\]
Note that the aforementioned leaf nodes form the terminal case in this definition. We additionally require candidate trees to be complete, i.e., the root $r$ of a tree $(g, r)$ generates all possible places $\mathit{cand}_g(r) = P$. In our setting, we may relax that to at least generating all through places, i.e., $P^! \subseteq \mathit{cand}_g(r) \subseteq P$. This is because non-through places can only have fitting behavior that does not intersect with them at all. Formally, $\set(\sigma) \cap ({\bullet}p \cup p{\bullet}) = \emptyset$ for place $p \in \mathbb{P}$ and behavior $\sigma \in \mathit{fit}(p)$. That makes them uninteresting in our setting.

\begin{figure}[htbp]
  \centering
  \includegraphics[width=0.8\linewidth]{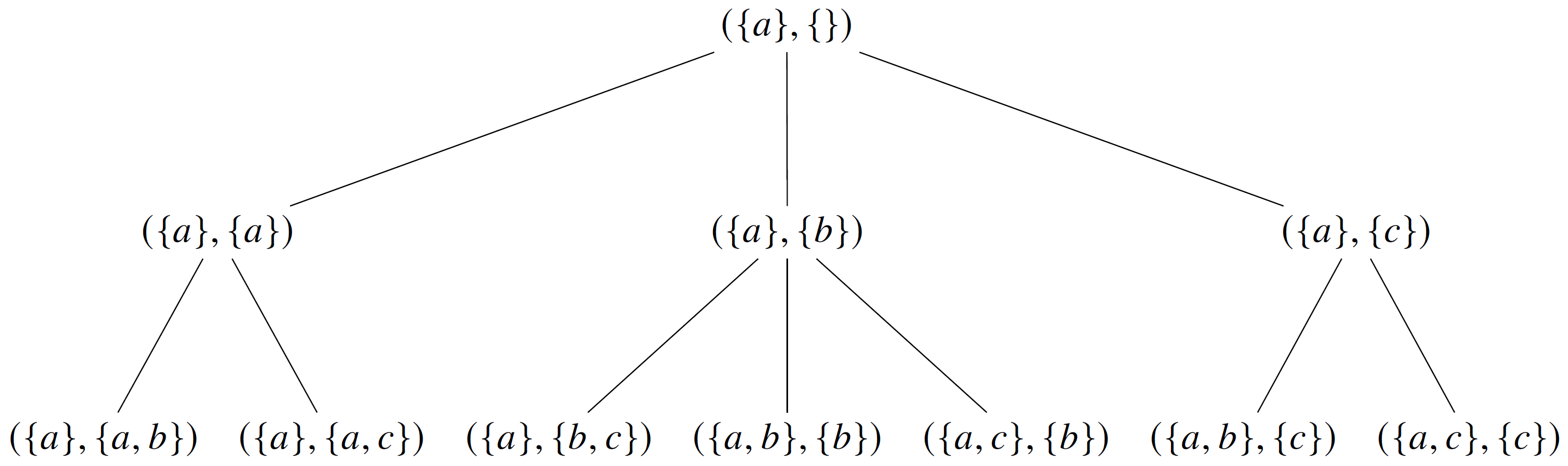}
  \caption{An exemplary \emph{partial} efficient candidate tree $T_\mathrm{ex} = (g, r)$ rooted at place $r = (\{a\}, \{\})$.}
  \label{fig:candidate-tree-example}
\end{figure}

Consider \Cref{fig:candidate-tree-example} for a partial example of such a tree $T_\mathrm{ex} = (g, r)$. Place $r = (\{a\}, \{\})$ is its root and it clearly fulfills the conditions of \Cref{def:candidate-tree}. Its candidate set $\mathit{cand}_g(r)$ consists of all visible places. Note that it is not complete.

The proposal step is simply a memory-efficient expansion of this tree. Potential candidates are proposed at most once. That is, only if they meet all constraints collected up to that moment. Constraints are how candidate evaluations and their indirect impact are incorporated into the candidate traversal.

\begin{definition}[Constraints]
\label{def:constraints}
A constraint $c : P \to \{\mathrm{True}, \mathrm{False}\}$ is a predicate function over the set of possible places. A place $p \in P$ meets constraint $c$, if and only if $c(p) = \mathrm{True}$.

To lift this to sets of constraints $C$, we use an indicator function $\mathbb{I}_C(p) :\Leftrightarrow \forall c \in C : c(p) = \mathrm{True}$. We use it to filter sets of places $S \subseteq P$ to those meeting the constraints $S \upharpoonright_{\mathbb{I}_C} = \{p \in S \mid \mathbb{I}_C(p) \text{ holds}\}$. For sets of places $S \subseteq P$, we also have $\mathbb{I}_C(S) :\Leftrightarrow \forall p \in S : \mathbb{I}_C(p)$ as an indicator whether all places meet all constraints. What we are really interested in are constraints that are aligned with the subtree relation of our candidate tree. This enables efficient pruning.
\end{definition}

\begin{definition}[Subtree-Monotonicity]
\label{def:subtree-monotonicity}
We call a set of constraints $C$ subtree-monotonic w.r.t.\ generation logic $g$ if for any $p, p' \in P$ with $p' \in \mathit{cand}_g(p)$,
\[
\neg\mathbb{I}_C(p) \implies \neg\mathbb{I}_C(p')
\]
In other words, if $C$ does not hold at $p$, it also does not hold on any place in its subtree. Note the contrapositive that if it holds at a descendant of $p$, it must also hold at $p$ itself.
\end{definition}

With respect to the generation logic of the example tree $T_\mathrm{ex}$, the set of constraints $C_\mathrm{ex} = \{p \mapsto a \notin {\bullet}p,\, p \mapsto a \notin p{\bullet}\}$, i.e., $a$ neither occurs in the preset nor postset of a place $p$, is subtree-monotonic.

A set of constraints $C$ induces a restriction on a generation logic $g$ as follows. For $p \in P$,
\[
g_C(p) \coloneqq g(p)\upharpoonright_{\mathbb{I}_C}.
\]
That is, children not meeting the constraints are filtered out. This point-wise restriction of the child generation logic suffices for respecting the constraints.

\begin{theorem}[Constrained Child Generation]
\label{thm:constrained-child-generation}
Let $C$ be a subtree-monotonic set of constraints with respect to a child generation logic $g$. For any $p \in P$, it holds that
\[
\mathit{cand}_g(p)\upharpoonright_{\mathbb{I}_C} =
\begin{cases}
\mathit{cand}_{g_C}(p) & \text{if } \mathbb{I}_C(p) \text{ holds,} \\
\emptyset & \text{otherwise.}
\end{cases}
\]
\end{theorem}

\Cref{thm:constrained-child-generation} states that the candidates meeting the constraints in the subtree rooted at $p$ are exactly the candidates that the restricted generation logic $g_C$ recursively generates. Simply said, a recursive traversal via $g_C$ does not miss any places meeting the constraints.

\begin{figure}[htbp]
  \centering
  \includegraphics[width=0.7\linewidth]{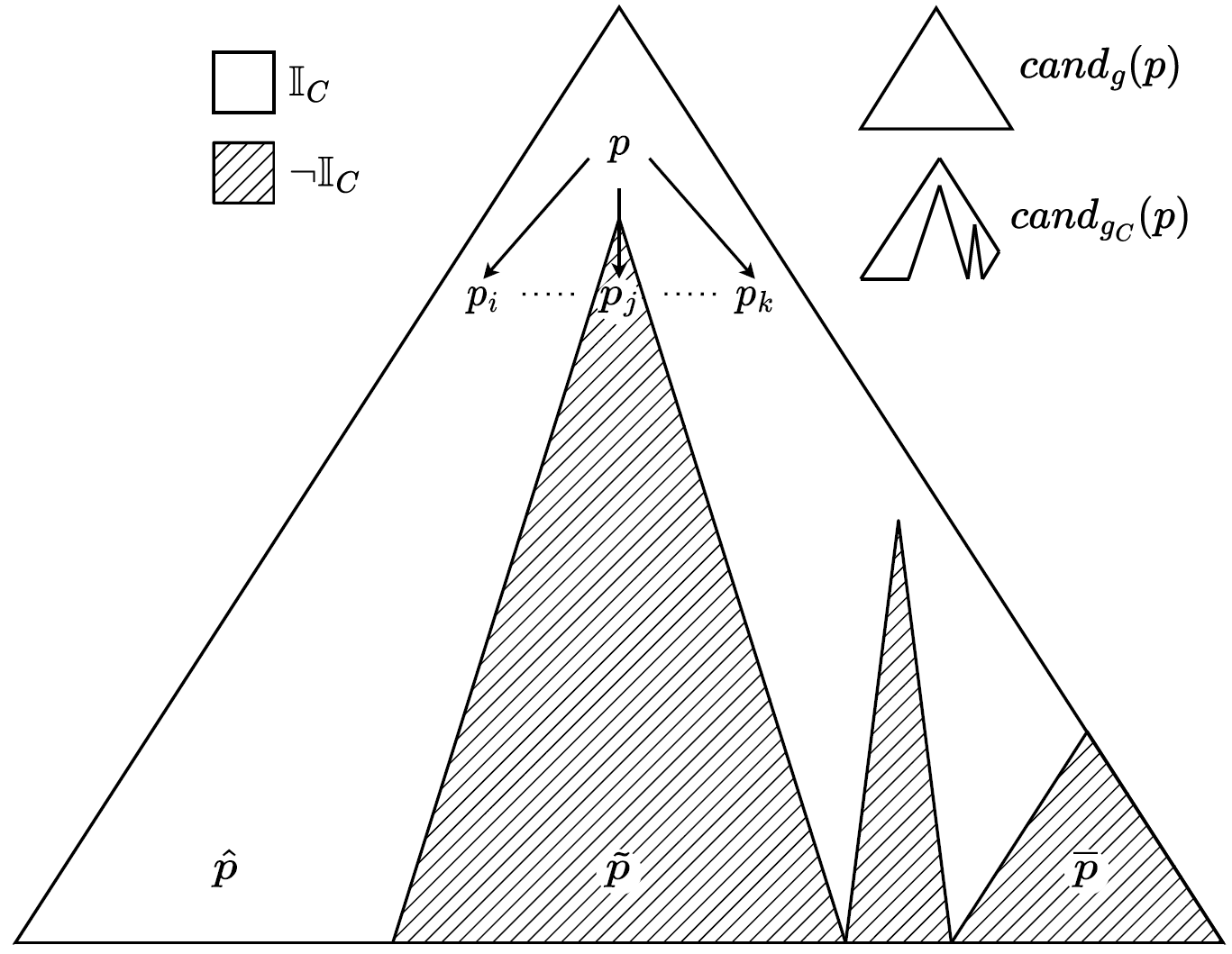}
  \caption{A schematic of how the subtree-monotonicity of a constraint set $C$ interacts with a child generation logic and resulting candidate set rooted at a place $p$. A place like $\hat{p}$ for which $\mathbb{I}_C$ holds, cannot have an ancestor for which it does not and is thus included in $\mathit{cand}_{g_C}(p)$. Places like $\tilde{p}$ or $\bar{p}$ for which $\mathbb{I}_C$ does not hold are either themselves directly filtered by $g_C$ or an ancestor of theirs is. Therefore, they are also not included in the generated candidate set $\mathit{cand}_{g_C}(p)$.}
  \label{fig:subtree-monotonicity}
\end{figure}

\begin{proof}
Let $p \in P$ be a place, $g$ a child generation logic and $C$ a subtree-monotonic set of constraints. If $\mathbb{I}_C(p)$ does not hold, $\mathit{cand}_g(p)\upharpoonright_{\mathbb{I}_C} = \{p' \in \mathit{cand}_g(p) \mid \mathbb{I}_C(p') \text{ holds}\} = \emptyset$ immediately follows from subtree-monotonicity.

So, assume $\mathbb{I}_C(p)$ holds. \Cref{fig:subtree-monotonicity} provides an illustration of the following proof. As $g_C$ is a point-wise restriction of $g$, i.e., $\forall p \in P : g_C(p) \subseteq g(p)$, we have $\mathit{cand}_{g_C}(p) \subseteq \mathit{cand}_g(C)$. Trivially, we also have $\mathit{cand}_g(p)\upharpoonright_{\mathbb{I}_C} \subseteq \mathit{cand}_g(p)$, so it suffices to show for any $\hat{p} \in \mathit{cand}_g(p) \setminus \{p\}$ that $\mathbb{I}_C(\hat{p}) \Leftrightarrow \hat{p} \in \mathit{cand}_{g_C}(p)$.

Let $\hat{p} \in \mathit{cand}_g(p) \setminus \{p\}$ be a descendant of $p$ and let $\langle p_1, \ldots, p_k \rangle \in P^*$, $k \in \mathbb{N}$ be the sequence (unique due to tree structure) of places from $p$ to $\hat{p}$, i.e., $p_1 = p$, $p_k = \hat{p}$ and $p_{i+1} \in g(p_i)$ for all $i \in \underline{k-1}$.

Assume $\mathbb{I}_C(\hat{p})$ holds. Due to the contrapositive of subtree-monotonicity, it also holds for all $p_i$ on that path, and thus $p_{i+1} \in g_C(p_i)$ for all $i \in \underline{k-1}$. We therefore get $\hat{p} \in \mathit{cand}_{g_C}(p)$.

Assume $\mathbb{I}_C(\hat{p})$ does not hold. Then, in particular, $\hat{p}$ is not generated by the restricted generation logic $g_C$ by its parent, i.e., $\hat{p} \notin g_C(p_{k-1})$. So, we have $\hat{p} \notin \mathit{cand}_{g_C}(p)$. 
\end{proof}

To allow general prioritized expansion, instead of hard-coding breadth- or depth-first traversal, we need to keep some additional state. The set of currently not fully expanded nodes as well as an indicator which of their children were already generated is necessary. Additionally, we keep a set of currently active subtree-monotonic constraints.

\begin{definition}[Tree State]
\label{def:tree-state}
A tree state $TS = (S, C, \nu)$ of a tree $T = (g, r)$ is characterized by a set of active nodes $S \subseteq P$, a set of constraints $C$ and a function $\nu : S \to P$ with $\forall s \in S : \nu(s) \subseteq g(s)$, indicating the already visited children of the active nodes.
\end{definition}

The canonical initial state of a candidate tree $T = (g, r)$ is $(\{r\}, \emptyset, \{(r, \emptyset)\})$. The expansion can now be guided by an expansion strategy that selects the next active node of a tree state to be expanded. Note that the overall candidate ordering is still mainly dependent on the tree structure, as children cannot be traversed before their parents.

\begin{definition}[Expansion Strategy]
\label{def:expansion-strategy}
An expansion strategy is a function $n : \mathcal{P}(P) \to P$ with $\forall \emptyset \neq S \subseteq P : n(S) \in S$ that chooses one place out of a non-empty set of places.
\end{definition}

A tree state can be pruned by removing active nodes (and their mapping) in case their constraint-filtered children are exhausted, i.e., they become inactive. This is highly dependent on the expansion strategy. In the worst case, the tree state can grow to the maximum tree width, as is the case with breadth-first traversal. Depth-first traversal only needs to keep track of one path from the root node to the current node, and is thus bounded by the tree depth. This makes it the most memory efficient traversal.

Another popular class of those strategies are heuristics, i.e., using a scoring function $h : P \to \mathbb{R}$. Without loss of generality, we assume a preference for smaller values. Then $h$ induces the strategy $n_h(S) = \operatorname{argmin}_{p \in S} h(p)$ for $\emptyset \neq S \subseteq P$. We call a $h$ subtree-monotonic w.r.t.\ generation logic $g$ if $\forall p, p' \in P : p' \in \mathit{cand}_g(p) \implies h(p) \leq h(p')$ holds. This is essentially a natural generalization of the constraint monotonicity from boolean implication to an ordering relation on the range of the heuristic. For example, the heuristic for place size, $h_\mathrm{ex} : p \mapsto |{\bullet}p \cup p{\bullet}|$ is subtree-monotonic w.r.t.\ the example $T_\mathrm{ex}$ (from \Cref{fig:candidate-tree-example}).

Subtree-monotonic heuristics also induce subtree-monotonic threshold constraints, e.g., for a heuristic $h$ and threshold $t \in \mathbb{R}$, $c_{h,t}(p) :\Leftrightarrow h(p) \leq t$ for $p \in P$. Note that, while we focus here on subtree-monotonic ordering and constraining strategies to guarantee completeness, more general strategies are possible and supported in our implementation. Dropping the completeness notion and letting greediness run wild can be empirically valid and is heavily employed on the composition side.

There is also another prioritization opportunity within the expansion step. Namely, the child ordering $m : P \to \bigcup_{S \subseteq P} \mathrm{Perm}(S)$ w.r.t.\ the generation logic $g$ that assigns to each place a permutation of its children under $g$, i.e., $m(p) = \pi \in \mathrm{Perm}(g(p))$ for $p \in P$. We use a static (place independent) child node permutation rule but it could be more sophisticated than that. For example, children with a heuristically determined high potential to improve the current result could be explored first.

First, the next node to expand is picked from the active set according to the expansion strategy. Then, the constrained set of potential unvisited children is computed. If it is empty, we recurse after removing this exhausted place. Otherwise, the first place according to the child ordering is added to the active set and returned.

With all of the previous definitions out of the way, we can define the tree expansion procedure in \Cref{alg:expand}. Expansion fails (returns $\bot$) if there are no more active nodes.

\begin{algorithm}[htbp]
\caption{\textsc{expand}}
\label{alg:expand}
\begin{algorithmic}[1]
\Require generation logic $g$
\Require expansion strategy $n$
\Require child ordering $m$
\Require tree state $TS = (S, C, \nu)$
\Require $S \neq \emptyset$
\Procedure{expand}{$TS$}
    \State $p \gets n(S)$
    \State $N \gets g_C(p) \setminus \nu(p)$
    \If{$N = \emptyset$} \Comment{pruning if exhausted due to constraints}
        \If{$|S| > 1$}
            \State \Return $\textsc{expand}((S \setminus \{p\}, C, \nu\upharpoonright_{S \setminus \{p\}}))$
        \Else
            \State \Return $\bot$
        \EndIf
    \Else
        \State $\pi \gets m(p)$
        \State $p' \gets \operatorname{argmin}_{\hat{p} \in N} \pi(\hat{p})$
        \State $\nu(p) \gets \nu(p) \cup \{p'\}$
        \If{$\nu(p) \supseteq g_C(p)$} \Comment{pruning if this was $p$'s last child}
            \State $S \gets S \setminus \{p\}$
            \State $\nu \gets \nu\upharpoonright_S$
        \EndIf
        \State $S \gets S \cup \{p'\}$
        \State $\nu \gets \nu \cup \{(p', \emptyset)\}$
        \State \Return $p'$
    \EndIf
\EndProcedure
\end{algorithmic}
\end{algorithm}

The tree expansion is called within the proposal routine \Cref{alg:propose} to generate all constraint-satisfying candidates. Whenever a place is produced, it is handed over to the composition side which we describe in \Cref{sec:composition}. Due to the acyclicity (tree-structure) of the child node generation function $g$ and children bookkeeping, \textsc{expand} never returns the same place twice when initialized with just the root node as an active node. This actually holds in general for sets of active nodes $S$ with disjoint subtrees, i.e., $\forall p \neq p' \in S : p \notin \mathit{cand}_g(p')$. Consequently, \textsc{propose} will necessarily terminate after a finite number of steps.

\begin{algorithm}[htbp]
\caption{\textsc{propose}}
\label{alg:propose}
\begin{algorithmic}[1]
\Require efficient tree $T = (g, r)$
\Require initial constraints $C$
\Procedure{propose}{}
    \State $TS \gets (\{r\}, C, \{(r, \emptyset)\})$
    \State $p \gets \begin{cases} r & \text{if } \mathbb{I}_C(r), \\ \bot & \text{otherwise.} \end{cases}$
    \While{$p \neq \bot$}
        \State \textbf{yield} $p$
        \State \textbf{update} $C$
        \State $p \gets \textsc{expand}(TS)$
    \EndWhile
\EndProcedure
\end{algorithmic}
\end{algorithm}

After each candidate proposal, we allow the constraint set to change. The idea is that new constraints may be generated in response to the candidate. We require that the set of constraints is only ever updated in a monotone fashion.

\begin{definition}[Monotonic Constraint Strengthening]
\label{def:monotonic-strengthening}
Let $C, C'$ be two subtree-monotonic constraint sets. $C'$ is a \emph{monotonic strengthening} of $C$ if and only if $\forall p \in P : \neg\mathbb{I}_C(p) \implies \neg\mathbb{I}_{C'}(p)$ holds, i.e., places for which $C$ does not hold, $C'$ does neither.
\end{definition}

Most trivially this is the case if $C \subseteq C'$. A sequence of constraint sets is monotonic if every two directly following constraint sets are monotonic strengthenings. This follows from the transitivity of $\implies$. For example, the sequence of constraint sets $\{\}$, $\{p \mapsto a \notin {\bullet}p \cap p{\bullet}\}$ ($a$ is not a self loop of $p$), $\{p \mapsto a \notin {\bullet}p\}$ ($a$ is not in the preset of $p$), $\{p \mapsto a \notin {\bullet}p \cup p{\bullet}\}$ ($a$ is neither in the preset nor postset of $p$) is monotonic.

We want to note here that there is a difference between the formal constraints and their representation. Very local types of constraints only affect the subtree rooted in a single place and therefore have to be checked only once. Some constraints ``touch'' many subtrees and have to be constantly re-evaluated while other types of constraints can be efficiently summarized and thus incorporated with marginal impact on performance. This is an important aspect to consider on both the conceptual as well as implementation side.

If we enforce the aforementioned monotonicity properties, we can guarantee a complete candidate traversal w.r.t.\ the final set of constraints $C$. If we further require the sequence of constraint sets to never exclude a previously returned candidate, the traversal is even minimal, i.e., we never consider a candidate that would later be filtered by new constraints. This is in particular the case for the local constraints mentioned above if they can be generated in response to the place they are local to.

\begin{proposition}
\label{prop:completeness-minimality}
Assuming constraint set updates are monotonic strengthenings, the set of returned places of the proposal routine \Cref{alg:propose}, denoted here by $M$, contains all places for which all constraints in the final constraint set $C$ hold, i.e., $M \supseteq P\upharpoonright_{\mathbb{I}_C}$ (completeness). Further, assuming an updated constraint set never excludes previously returned places, no more than those places are returned, i.e., $M = P\upharpoonright_{\mathbb{I}_C}$ (minimality).
\end{proposition}

\begin{proof}
Let $(g, r)$ be an efficient candidate tree and $C_0$ be the initial constraints. Assume \textsc{propose} terminates after $T \in \mathbb{N}_0$ executions of the loop (steps). For $T = 0$, we never return anything and as $C_0$ is subtree-monotonic, $\neg\mathbb{I}_{C_0}(r) \implies P\upharpoonright_{\mathbb{I}_{C_0}} = \emptyset$. Let $M_t$ be the set of returned candidates and $C_t$ be updated constraint set after the $t$-th step for $t \in \underline{T}$.

We first show completeness by contradiction. Assume $p \in P\upharpoonright_{\mathbb{I}_{C_T}} \setminus M_T$ exists. As the constraint sequence is monotonic, the set of places meeting all constraints can only shrink, i.e., $\forall t < T : P\upharpoonright_{\mathbb{I}_{C_t}} \supseteq P\upharpoonright_{\mathbb{I}_{C_{t+1}}}$. So, for all $t \in \underline{T}$, $p \in P\upharpoonright_{\mathbb{I}_{C_t}}$ holds. It can thus have never been filtered directly by the restriction of the child generation logic $g_{C_t}(p) = g(p)\upharpoonright_{\mathbb{I}_{C_t}}$ at any timestep $t$. $p$ also cannot be the root $r$ as that is returned exactly if $\mathbb{I}_{C_0}(r)$ holds. Then assume an ancestor $p'$, i.e., $p \in \mathit{cand}_g(p')$, was filtered at step $0 \leq k \leq T$ instead. So, $\neg\mathbb{I}_{C_k}(p')$. However, as $C_k$ is subtree-monotonic, this directly implies $\neg\mathbb{I}_{C_k}(p)$, i.e., $p \notin P\upharpoonright_{\mathbb{I}_{C_k}}$. A contradiction.

We show minimality $M_t \subseteq P\upharpoonright_{\mathbb{I}_{C_{t-1}}}$ via induction over the timesteps. For $t = 1$, $M_1 = \{r\} \subseteq P\upharpoonright_{\mathbb{I}_{C_0}}$ holds as we enter the loop the first time, i.e., $\mathbb{I}_{C_0}(r) = \mathrm{True}$. For $t + 1 > 1$, $M_{t+1} = M_t \cup \{p\}$ where $p$ is the result of the \textsc{expand} call from the previous iteration. By construction, we have $\mathbb{I}_{C_t}(p) = \mathbb{I}_{C_{t-1}}(M_t) = \mathrm{True}$. If $C_t$ does not constrain any place in $M_t$, i.e., $\mathbb{I}_{C_t}(M_t)$ holds, then it directly follows that $\mathbb{I}_{C_t}(M_t \cup \{p\}) = \mathrm{True}$. So, $M_{t+1}$ is still a subset of all places meeting the updated constraints, i.e., $M_{t+1} \subseteq P\upharpoonright_{\mathbb{I}_{C_t}}$.

At the end after $T$ steps, we thus inductively get $M_T \subseteq P\upharpoonright_{\mathbb{I}_{C_{T-1}}}$. Finally, as the constraint sequence is monotonic and the algorithm terminated, there were no more unseen candidates left under $C_T$. Together with the assumption that it does not constrain an already returned candidate, we have exactly $M_T = P\upharpoonright_{\mathbb{I}_{C_T}}$, i.e., the set of proposed places is complete and minimal. 
\end{proof}

This motivates choosing a good constraining strategy that generates constraints early enough for the selected generation logic. For example, consider the constraint sequence $C_0 = \emptyset, C_1 = \ldots = C_T = \{p \mapsto a \notin p{\bullet}\}$ (where \textsc{propose} terminated after $T$ steps), i.e., the constraint that $a$ is not in the postset of a place $p$ is added after the first proposed place. On a breadth-first traversal of the partial example tree $T_\mathrm{ex}$, the left subtree (at $p = (\{a\}, \{a\})$) would be completely pruned.

\subsubsection{Concrete Child Generation Logic}
\label{sec:concrete-child-generation}

In this section, we give our instantiation of a child generation logic. As the number of activities is the exponential factor in the runtime of this class of algorithms, we want to provide more fine-grained control over the selection of relevant activities. For this section, we thus fix a set of considered preset activities $\emptyset \neq A_\mathrm{pre} \subseteq A$ and postset activities $\emptyset \neq A_\mathrm{post} \subseteq A$, making the set of possible places $P = \mathcal{P}(A_\mathrm{pre}) \times \mathcal{P}(A_\mathrm{post})$. Among these, only the through places $P^! = (\mathcal{P}(A_\mathrm{pre}) \setminus \{\emptyset\}) \times (\mathcal{P}(A_\mathrm{post}) \setminus \{\emptyset\})$ are interesting. This allows us to naturally exclude places with the designated start/end activities $\eactivity$ in their preset and $\sactivity$ in their postset which can never have any fitting behavior.

The candidate tree is rooted at the empty place, i.e., $r = (\emptyset, \emptyset)$. It is grown by extending the preset and postset until all considered activities are covered. That makes the depth of a place $p \in P$ equal to the number of connected arcs, $|{\bullet}p| + |p{\bullet}|$. This intuitive measure of place simplicity thus becomes a first class property in the emerging candidate ordering.

\begin{definition}[Incremental Pre- \& Postset Expansions]
\label{def:incremental-expansions}
Let $p \in P$ be a place and $<_\mathrm{pre}$, $<_\mathrm{post}$ be strict total orderings on $A_\mathrm{pre}$ and $A_\mathrm{post}$ respectively. An incremental preset (postset) expansion of $p = (I, O)$ is an extension of $p$ with one activity added to its preset (postset) which is larger than previously contained activities according to $<_\mathrm{pre}$ (or $<_\mathrm{post}$ respectively). Formally,
\begin{align*}
\mathit{pre}^+(p) &= \{(I \cup \{a\}, O) \mid a \in A_\mathrm{pre},\, \forall a' \in I : a' <_\mathrm{pre} a\} \\
\mathit{post}^+(p) &= \{(I, O \cup \{a\}) \mid a \in A_\mathrm{post},\, \forall a' \in O : a' <_\mathrm{post} a\}
\end{align*}
\end{definition}

Note that if $I = \emptyset$ (or $O = \emptyset$, respectively), all considered activities meet the ordering constraint which is intuitively correct, as any element is then a valid extension of the empty set. The concrete child generation logic then simply concatenates these expansions under a condition on postset size to preserve the one-parent tree structure. Without it, we would have non-unique parents, as a through place is both a preset and a postset expansion of two different places. The choice of the postset size condition is related to pruning efficiency but could symmetrically be on the preset. For $p = (I, O) \in P$,
\[
g(p) \coloneqq
\begin{cases}
\mathit{post}^+(p) & \text{if } I = O = \emptyset, \\
\mathit{pre}^+(p) & \text{if } |I| + |O| = 1, \\
\mathit{post}^+(p) \cup \mathit{pre}^+(p) & \text{if } |I| \geq 1 \wedge |O| = 1, \\
\mathit{post}^+(p) & \text{if } |I| \geq 1 \wedge |O| > 1
\end{cases}
\]
These cases are disjoint and well defined on all places eventually generated by $g$ starting at the empty root $r = (\emptyset, \emptyset)$. It is important to note here that $g$ does not generate non-through places except for in the first two layers, where all places have either ${\bullet}p = \emptyset$ or $p{\bullet} = \emptyset$. This can also be modeled as an initial constraint $c(p) :\Leftrightarrow ({\bullet}p = p{\bullet} = \emptyset) \vee ({\bullet}p = \emptyset \wedge |p{\bullet}| = 1) \vee ({\bullet}p \neq \emptyset \wedge p{\bullet} \neq \emptyset)$ that exactly excludes $P \setminus \mathit{cand}_g(r)$, leaving $\mathit{cand}_g(r) = P^! \cup \{r\} \cup \mathit{post}^+(r)$.

Our child ordering permutation $m$ merely lifts $<_\mathrm{pre}$ and $<_\mathrm{post}$ to the extensions constructed with them. That is, let $m(p) = \pi$, then $\pi((I \cup \{a\}, O)) < \pi((I \cup \{a'\}, O)) :\Leftrightarrow a <_\mathrm{pre} a'$ and $\pi((I, O \cup \{a\})) < \pi((I, O \cup \{a'\})) :\Leftrightarrow a <_\mathrm{post} a'$ hold. Further, postset expansions are ordered before preset expansions, i.e., $\forall p_I \in \mathit{pre}^+(p), p_O \in \mathit{post}^+(p) : \pi(p_O) < \pi(p_I)$, because they can lead to earlier pruning. This is related to the choice of the condition on the postset size above.

An example candidate subtree according to the instantiations given in this section is presented in \Cref{fig:candidate-subtree}. It is rooted at $(\{\sactivity\}, \{a\})$, so it starts at depth two. Blue edges indicate postset expansions, red ones preset expansions. The activity orderings used in $g$ are $\sactivity <_\mathrm{pre} a <_\mathrm{pre} b <_\mathrm{pre} c$ and $a <_\mathrm{post} b <_\mathrm{post} c <_\mathrm{post} \eactivity$. Note that the selected sets of possible activities for the preset $A_\mathrm{pre}$ and postset $A_\mathrm{post}$ differ. Any place $p$ with $\eactivity \in {\bullet}p$ and $\sactivity \in p{\bullet}$ can never have any fitting behavior, so they can readily be excluded from the candidate tree. This example subtree illustrates the emerging asymmetrical structure with the earlier and more homogeneous postset expansion subtrees. The over/underfedness constraints can cut off child subtrees in which only one of the two colors occurs. So, this postset (blue) first structure prioritizes underfedness. Constraint generation enabling pruning is enabled by candidate evaluation and performed within the composition logic.

\begin{figure}[p]
  \centering
  \rotatebox{90}{%
    \begin{minipage}{\textheight}
      \centering
      \includegraphics[
        width=\linewidth,
        height=\textwidth,
        keepaspectratio
      ]{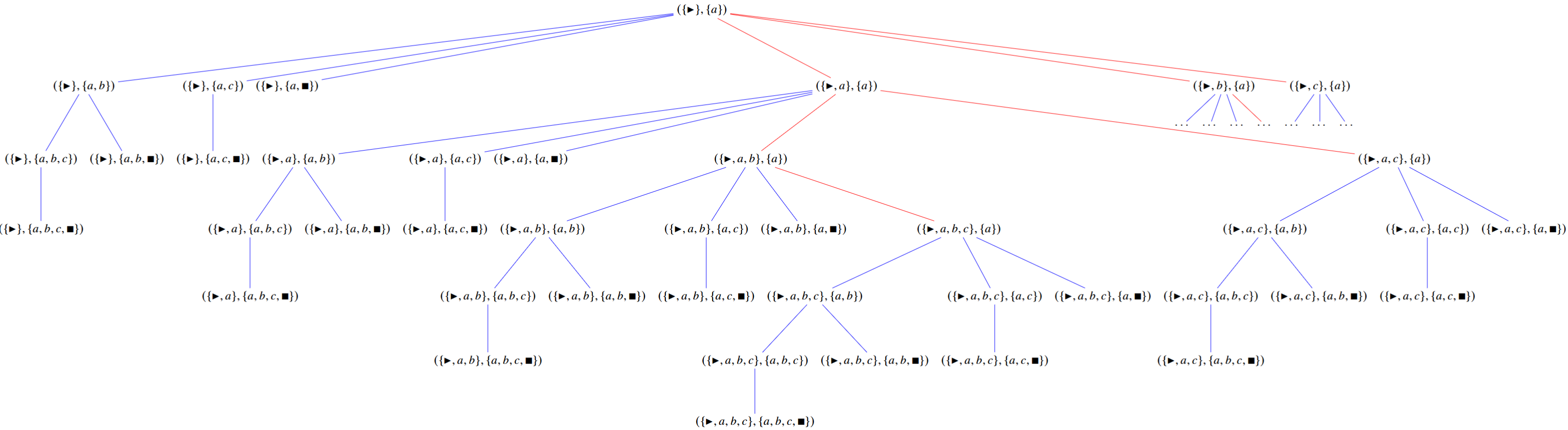}

      \captionof{figure}{An example of the candidate subtree rooted at $(\{\sactivity\}, \{a\})$ generated by our given child generation logic $g$, child permutation $m$ and activity orderings $\sactivity <_\mathrm{pre} a <_\mathrm{pre} b <_\mathrm{pre} c$ and $a <_\mathrm{post} b <_\mathrm{post} c <_\mathrm{post} \eactivity$. Blue edges indicate postset expansions, red ones preset expansions. The emerging structure has more homogeneous postset expansion subtrees which are also ordered earlier (according to breadth-first traversal).}
      \label{fig:candidate-subtree}
    \end{minipage}%
  }
\end{figure}

\subsection{Evaluation}
\label{sec:evaluation}

Any proposed candidate place is evaluated with arbitrary metrics to provide a basis for the composition logic. We divide the evaluations of candidates into \emph{individual} and \emph{relative}. Individual evaluations are defined on individual places, i.e., $\mathit{eval}^\mathit{ind} : \mathbb{P} \to E^\mathit{ind}$ with evaluation co-domain $E^\mathit{ind}$, whereas relative evaluations are always relative to a given set of places, i.e., $\mathit{eval}^\mathit{rel} : \mathbb{P} \times \mathcal{P}(\mathbb{P}) \to E^\mathit{rel}$ with evaluation co-domain $E^\mathit{rel}$. Individual evaluations are computable locally, i.e., independently of the previously proposed candidates. They may use static parameters, such as the input event log, however. Typical for this type are log-fitness evaluations (based on place fitness from \Cref{def:fitness}) and generally place-local heuristics.

Relative evaluations on the other hand explicitly consider a candidate place together with an intermediate result (set of places). This allows testing model-level/global properties, such as deadlock-freedom of the resulting Petri net or place implicitness. These functions tend to be more expensive to compute due to being defined on the (place) set level, as they may scale super-linearly or even exponentially in the number of elements.

\subsection{Concrete Evaluations}
\label{sec:concrete-evaluations}

In the following, we present two types of concrete evaluations used in our instantiation of the framework: place fitness and place implicitness.

\paragraph{Place Fitness}

Recall the definition of $\square_\sigma(p)$ on individual places and place sets (\Cref{def:fitness}). If we want the resulting set of places to be fitting for the entire input log, we can translate that into the constraint that every individual place needs to be fitting for the entire log. Closely related to the tree structure we gave before is the notion of under- and overfedness extensively discussed by van der Aalst in~\cite{vanderAalst2018DiscoveringGlue}.

\begin{definition}[Under- \& Overfedness]
\label{def:underfedness}
Let $p \in \mathbb{P}$ be a place and $\sigma \in \mathbb{B}$ a behavior. We call $p$ \emph{overfed} (\emph{underfed}) on $\sigma$ if and only if $\triangle_\sigma(p)$ (or $\triangledown_\sigma(p)$, respectively), with
\begin{align*}
\triangle_\sigma(p) &:\Leftrightarrow \neg\,\text{\textsc{balanced}}(\sigma, p) \\
\triangledown_\sigma(p) &:\Leftrightarrow \neg\,\text{\textsc{non-neg}}(\sigma, p)
\end{align*}
\end{definition}

Note that these properties are not mutually exclusive. They are also not entirely symmetric. Overfedness corresponds to a non-empty marking at the end of token-based replay, i.e., an excess of token production, and underfedness corresponds to a negative marking during token-based replay, i.e., a deficiency in token production. The relation to fitness $\square_\sigma(p)$ is trivially $\square_\sigma(p) \Leftrightarrow \neg\triangle_\sigma(p) \wedge \neg\triangledown_\sigma(p)$.

For place evaluation, we are interested in the event log-level. Luckily, behavior-level properties can be lifted to multisets of behaviors via aggregation by counting.

\begin{definition}[Multiset-level Lifting]
\label{def:multiset-lifting}
For a non-empty event log $[\,] \neq L \in \mathcal{B}(\mathbb{B})$, a place $p \in \mathbb{P}$ and a behavior-level property $R : \mathbb{B} \times \mathbb{P} \to \{\mathrm{True}, \mathrm{False}\}$ on places and behaviors we can count the fraction of behaviors it holds on as follows.
\[
\#^R_{L,p} = \frac{|[\sigma \in L \mid R(\sigma, p) = \mathrm{True}]|}{|L|}
\]
\end{definition}

This also lifts behavior-level monotonicity to the log-level.

\begin{theorem}[Lifted Monotonicity]
\label{thm:lifted-monotonicity}
Let $[\,] \neq L \in \mathcal{B}(\mathbb{B})$ be a non-empty event log, $p \in \mathbb{P}$ be a place, and $R : \mathbb{B} \times \mathbb{P} \to \{\mathrm{True}, \mathrm{False}\}$ be a behavior-level property. For any $p' \in \mathbb{P}$ where $R(\sigma, p) \implies R(\sigma, p')$ holds for all $\sigma \in L$, the following holds.
\[
\#^R_{L,p} \leq \#^R_{L,p'}
\]
\end{theorem}

\begin{proof}
Every behavior counted in the denominator of the left hand side of the inequality is also counted on the right, as $R(\sigma, p) \implies R(\sigma, p')$ for all $\sigma \in L$.
\end{proof}

In other words, if a behavior-level property $R$ is monotonic with respect to some ordering on places, e.g., the child generation function of our candidate tree, the fractions will also be monotonic. This will be the basis of our constraint generation which we introduce in \Cref{sec:concrete-constraint-generation}.

We define a fitness evaluator that counts the fractions of fitting, underfed, and overfed behaviors.
\begin{align*}
\mathit{eval}^\mathit{ind}_{\square,\triangledownscript,\triangle} : \mathbb{P} &\to [0,1] \times [0,1] \times [0,1] \\
p &\mapsto (\#^\square_{L,p},\, \#^{\triangledownscript}_{L,p},\, \#^\triangle_{L,p})
\end{align*}

\paragraph{Place Implicitness}

The implicitness evaluation checks whether a place is redundant, i.e., does not change the fitting behavior of an intermediate result. An additional place can of course at most reduce the fitting behavior, and never increase it. Adding an implicit place to the solution set would just increase the complexity of the resulting model, so it is beneficial to filter it out. The corresponding evaluator can look as follows.
\begin{align*}
\mathit{eval}^\mathit{rel}_\text{global implicitness} : \mathbb{P} \times \mathcal{P}(\mathbb{P}) &\to \{\mathrm{True}, \mathrm{False}\} \\
(p, S) &\mapsto \mathit{fit}(S) = \mathit{fit}(\{p\} \cup S)
\end{align*}
This property can be computed via linear programming based on Petri net theory~\cite{GarciaValles1999ImplicitPlaces}. However, this is relatively expensive. Given a fixed set of behaviors $B \subseteq \mathbb{B}$, e.g., those occurring within an input event log, we can check the implicitness somewhat heuristically by only considering whether $p$ is redundant on this sample of all possible behavior:
\begin{align*}
\mathit{eval}^\mathit{rel}_\text{replay implicitness} : \mathbb{P} \times \mathcal{P}(\mathbb{P}) &\to \{\mathrm{True}, \mathrm{False}\} \\
(p, S) &\mapsto \mathit{fit}(S) \cap B = \mathit{fit}(\{p\} \cup S) \cap B
\end{align*}
This implicit place filtering also has huge computational advantages as it restricts the solution set space, i.e., all combinations of possible places, dramatically and keeps the intermediate result small, thus speeding up all relative evaluations. Mannel et al.\ introduced this technique as ``replay-based implicit place removal''. For more details, refer to~\cite{Mannel2020RemovingImplicit}.

To facilitate efficient application of a suite of evaluators, it is up to the composing side to dynamically call them as needed. This allows for short-circuiting of composition logic, e.g., discarding a candidate place after it does not meet a particular threshold on a cheap-to-calculate heuristic. The following section introduces the generalized composition logic, and our greedy instantiation.

\subsection{Composition}
\label{sec:composition}

The handling of proposed candidates and the constraining strategy, i.e., the system by which to generate constraints, is implemented on the composing side. Constraints usually derive from desired properties of the possible solution sets and are discovered by evaluation of candidate places.

In general, our framework is intended to support arbitrary composition strategies that internally update a state, making use of the aforementioned evaluations. In particular, we see a lot of opportunities in mirroring the efficient candidate traversal of the proposal side. We give the general structure in \Cref{alg:general-composing}. The procedure \textsc{compose} is called after \textsc{propose} yields a place. The candidate place is incorporated into an internal state by \textsc{update-state}, then, if possible, candidate constraints are generated. Only after calling \textsc{get-result} is the state collapsed to a concrete resulting set of places. However, with existing approaches, we are limited to greedy or at most relaxed greedy approaches to curb the huge exponential blowup of candidate combinations. Essentially, the composition routine turns into an online algorithm that has to make local decisions and trust the candidate generation process with a ``good'' ordering of candidates. \Cref{alg:greedy-composing} shows this simplified interface implementation. Instead of an arbitrary state, we directly manage a result set. The core strategy, in addition to \textsc{generate-constraints} as before, now lies in \textsc{deliberate-acceptance} which determines the manner in which the evaluated candidate is incorporated into the intermediate result.

In the following, we describe the quintessential constraint generation strategy that fits nicely into this online setting. Secondly, we give an overview of composition state updating.

\begin{algorithm}[htbp]
\caption{\textsc{general composing}}
\label{alg:general-composing}
\begin{algorithmic}[1]
\Require current state $S \in SP$, with state-space $SP$
\Require individual evaluators $(\mathit{eval}^\mathit{ind}_k)_{k \in N^\mathit{ind}}$, $N^\mathit{ind} = \underline{n^\mathit{ind}}$
\Require relative evaluators $(\mathit{eval}^\mathit{rel}_k)_{k \in N^\mathit{rel}}$, $N^\mathit{rel} = \underline{n^\mathit{rel}}$
\Require subroutine \textsc{evaluate}
\Require subroutine \textsc{update-state}
\Require subroutine \textsc{generate-constraints}
\Require subroutine \textsc{generate-result}: $SP \to \mathcal{P}(\mathbb{P})$
\Procedure{compose}{$p$}
    \State $\mathit{evs} \gets \textsc{evaluate}(p,\, (\mathit{e}^\mathit{ind}_k)_{k \in N^\mathit{ind}},\, (\mathit{e}^\mathit{rel}_k)_{k \in N^\mathit{rel}}p)$
    \State \textsc{update-state}$(p, \mathit{evs})$
    \State $C \gets \textsc{generate-constraints}(p, \mathit{evs})$
    \State \textbf{yield} constraints $C$
\EndProcedure
\Procedure{get-result}{}
    \State \Return \textsc{generate-result}$(S)$
\EndProcedure
\end{algorithmic}
\end{algorithm}

\begin{algorithm}[htbp]
\caption{\textsc{greedy composing}}
\label{alg:greedy-composing}
\begin{algorithmic}[1]
\Require set of previously accepted places $I \subseteq \mathbb{P}$
\Require subroutine \textsc{deliberate-acceptance}
\Procedure{evaluate}{$p,\, (\mathit{eval}^\mathit{ind}_k)_{k \in N^\mathit{ind}},\, (\mathit{eval}^\mathit{rel}_k)_{k \in N^\mathit{rel}}$}
    \For{$k \in N^\mathit{ind}$}
        \State $e^\mathit{ind}_k \gets \mathit{eval}^\mathit{ind}_k(p)$
    \EndFor
    \For{$k \in N^\mathit{rel}$}
        \State $e^\mathit{rel}_k \gets \mathit{eval}^\mathit{rel}_k(p, I)$
    \EndFor
    \State \Return $(e^\mathit{ind}_k)_{k \in N^\mathit{ind}},\, (e^\mathit{rel}_k)_{k \in N^\mathit{rel}})$
\EndProcedure
\Procedure{update-state}{$p, \mathit{evs}$}
    \State $\mathit{decision} \gets \textsc{deliberate-acceptance}(p, \mathit{evs})$
    \If{$\mathit{decision} = \mathrm{Accept}$}
        \State $I \gets I \cup \{p\}$
    \ElsIf{$\mathit{decision} = \mathrm{ReplaceExisting}(\hat{p})$}
        \State $I \gets I \setminus \{\hat{p}\} \cup \{p\}$
    \ElsIf{$\mathit{decision} = \mathrm{Reject}$}
        \State $I \gets I$
    \EndIf
\EndProcedure
\Procedure{generate-result}{$I$}
    \State \Return $I$
\EndProcedure
\end{algorithmic}
\end{algorithm}

\subsubsection{Concrete Constraint Generation}
\label{sec:concrete-constraint-generation}

As we have already mentioned in the proposal and evaluation sections, simple fitness constraints drive the constraint generation. The final connection to make from the previously introduced under/overfedness notions is to the preset and postset expansions used in our child generation logic $g$, which we fix for the following considerations. We first state the central theorem which expresses how these fitness notions align to our tree structure.

\begin{theorem}[Under/Overfedness Monotonicity]
\label{thm:underfedness-monotonicity}
Let $p, p' \in \mathbb{P}$ be places such that $p'$ is a postset expansion of $p$, i.e., ${\bullet}p = {\bullet}p'$ and $p{\bullet} \subset p'{\bullet}$. For any behavior $\sigma \in \mathbb{B}$, it holds that $\triangledown_\sigma(p) \implies \triangledown_\sigma(p')$. This is analogous for $\triangle$ and preset expansions.
\end{theorem}

In words, if $p$ is underfed (overfed) on $\sigma$, any postset (preset) expansion of $p$ is also underfed (overfed) on $\sigma$. Intuitively, a place that is underfed can only become more underfed by adding token consumers. As a concrete example, consider behavior $\sigma = \langle \sactivity, b, \eactivity \rangle$ and place $p = (\{a\}, \{b\})$. $p$ is underfed on $\sigma$ and any postset expansion of $p$, e.g., $(\{a\}, \{b, c\})$, must also be underfed on $\sigma$.

To be able to generate subtree-monotonic constraints from this, they have to be consistent on all places in the subtree rooted at some place $p$, i.e., the entirety of $\mathit{cand}_g(p)$.

\begin{lemma}[Underfedness Subtree-Monotonicity]
\label{lem:underfedness-subtree}
Let $p \in \mathbb{P}$ with $|{\bullet}p| + |p{\bullet}| \geq 2$ be a place of depth at least two and $p' \in \mathit{post}^+(p)$ be a direct postset expansion of it. For any $\sigma \in \mathbb{B}$, it holds that:
\[
\triangledown_\sigma(p) \implies \forall \hat{p} \in \mathit{cand}_g(p') : \triangledown_\sigma(\hat{p})
\]
\end{lemma}

\begin{proof}
To see this, recall the case distinction in our definition of $g$. If $p'$ is a postset expansion of $p$ with depth at least two and thus $|p{\bullet}| \geq 1$, we have $|p'{\bullet}| > 1$. Which implies that there are no preset expansions in $\mathit{cand}_g(p')$ as we never shrink places. Formally, $\forall \hat{p} \in \mathit{cand}_g(p') : {\bullet}p = {\bullet}p' = {\bullet}\hat{p} \wedge p{\bullet} \subset p'{\bullet} \subseteq \hat{p}{\bullet}$. As all places in $\mathit{cand}_g(p')$ are postset expansions of $p'$, \Cref{thm:underfedness-monotonicity} directly implies the claim.
\end{proof}

Due to our design decision to prefer postset expansions in our definition of $g$, the overfedness subtree-monotonicity is a lot more restricted.

\begin{lemma}[Overfedness Subtree-Monotonicity]
\label{lem:overfedness-subtree}
Let $p \in \mathbb{P}$ be a place with $\mathit{post}^+(p) = \emptyset$ and $p' \in \mathit{pre}^+(p)$ be a direct preset expansion of it. For any $\sigma \in \mathbb{B}$, it holds that:
\[
\triangle_\sigma(p) \implies \forall \hat{p} \in \mathit{cand}_g(p') : \triangle_\sigma(\hat{p})
\]
\end{lemma}

\begin{proof}
A place $p$ with non-expandable postset, i.e., $\mathit{post}^+(p) = \emptyset$, can only have preset expansions in its subtree. The claim then follows directly from \Cref{thm:underfedness-monotonicity}.
\end{proof}

Note that the only places $p$ in an unconstrained tree that have a non-expandable postset but still have preset children to prune are those with $p{\bullet} = \{x\}$ where $x \in A_\mathrm{post}$ is the maximal element of $A_\mathrm{post}$, i.e., $\forall a \in A_\mathrm{post} \setminus \{x\} : a <_\mathrm{post} x$.

We showed that these behavior-level properties $\triangledown$ and $\triangle$ satisfy the desired monotonicity. Using \Cref{thm:lifted-monotonicity}, we can thus conclude the following property of the counted fractions.

\begin{lemma}[Lifted Under/Overfedness Monotonicity]
\label{lem:lifted-underfedness}
Let $[\,] \neq L \in \mathcal{B}(\mathbb{B})$ be a non-empty event log and $p \in \mathbb{P}$ be a place of depth at least 2. For all places in the candidate sets of its postset expansions, i.e., $\bigcup_{p' \in \mathit{post}^+(p)} \mathit{cand}_g(p')$, it holds that:
\[
\#^{\triangledownscript}_{L,p} \leq \#^{\triangledownscript}_{L,p'}
\]
This is analogous for $\triangle$ and preset expansions.
\end{lemma}

\begin{proof}
Use \Cref{thm:lifted-monotonicity} with the $\triangledown$ as the behavior-level property, which is monotonic for the considered places according to \Cref{lem:underfedness-subtree}.
\end{proof}

We can now use the log-level lifting of the place fitness evaluator $\mathit{eval}^\mathit{ind}_{\square,{\triangledownscript},\triangle}$ to filter places and generate constraints. Specifically, we use thresholds to convert the fractions back again into truth values. It is clear to see that $\square_L(p) \Leftrightarrow \#^\square_{L,p} = 1$. Additionally, we have the following relation to over- and underfedness: $\#^\square_{L,p} \leq 1 - \#^{\triangledownscript}_{L,p}$ and $\#^\square_{L,p} \leq 1 - \#^\triangle_{L,p}$. That is, if a place should fit a certain fraction $\tau \in [0, 1]$ of behavior in a log, it can be underfed (and overfed) on at most $(1 - \tau) \cdot 100\%$ of the traces.

Given a threshold $\tau \in [0, 1]$ and an evaluation result $(\#^\square_{L,p}, \#^{\triangledownscript}_{L,p}, \#^\triangle_{L,p})$ of a place $p$ on log $L$, we can then generate the following constraints in \textsc{generate-constraint}.

If $\#^{\triangledownscript}_{L,p} > 1 - \tau \wedge |{\bullet}p| + |p{\bullet}| \geq 2$ holds, return constraint
\[
\text{\textsc{prune-postset-subtree}}_p(\hat{p}) :\Leftrightarrow \hat{p} \notin \bigcup_{p' \in \mathit{post}^+(p)} \mathit{cand}_g(p')
\]
Otherwise, if $\#^\triangle_{L,p} > 1 - \tau \wedge \mathit{post}^+(p) = \emptyset$ holds, return constraint
\[
\text{\textsc{prune-preset-subtree}}_p(\hat{p}) :\Leftrightarrow \hat{p} \notin \bigcup_{p' \in \mathit{pre}^+(p)} \mathit{cand}_g(p')
\]
These constraints might seem expensive to check, however, due to our choice of $g$, they only affect the subtrees of $p$, i.e., they are local. Thus, they can trivially be represented such that $g(p)\upharpoonright_{\mathbb{I}_\text{\textsc{prune-postset-subtree}}} = g(p) \setminus \mathit{post}^+(p)$, or respectively $g(p)\upharpoonright_{\mathbb{I}_\text{\textsc{prune-preset-subtree}}} = g(p) \setminus \mathit{pre}^+(p)$, takes no additional checks by simply marking these children as already visited in the tree state via $\nu$.

To visualize this, consider the candidate subtree presented in \Cref{fig:candidate-subtree} and event log $L = [\langle \sactivity, a, b, \eactivity \rangle^4, \langle \sactivity, a, a, \eactivity \rangle^6]$. Assume the root place $r = (\{\sactivity\}, \{a\})$ is being proposed. On behavior $\langle \sactivity, a, b, \eactivity \rangle$, $r$ is fitting, not underfed, and not overfed, and on $\langle \sactivity, a, a, \eactivity \rangle$, it is not fitting, underfed, and not overfed. Thus, $\#^\square_{L,r} = 0.4$, $\#^{\triangledownscript}_{L,r} = 0.6$ and $\#^\triangle_{L,r} = 0$. Given the threshold $\tau = 0.5$, we can discard $r$, as no Petri net that contains $r$ could replay more than 40\% of $L$. Furthermore, we can generate the constraint $\text{\textsc{prune-postset-subtree}}_r(\hat{p}) :\Leftrightarrow \hat{p} \notin \bigcup_{p' \in \mathit{post}^+(r)} \mathit{cand}_g(p')$. That means, we can prune all the descendants of $r$ that are connected purely along blue edges (postset expansions). In this case $(\{\sactivity\}, \{a, b\})$, $(\{\sactivity\}, \{a, c\})$, $(\{\sactivity\}, \{a, \eactivity\})$ and their children (left side of the tree). We can easily verify that none of the affected places $p$ could ever have an underfed fraction smaller than 0.6, i.e., we know $\#^{\triangledownscript}_{L,p} \geq 0.6$ and thus $\#^\square_{L,p} \leq 1 - \#^{\triangledownscript}_{L,p} \leq 0.4$. At the same time, we can see that places that are descended via a red edge (preset expansion), e.g., $p' = (\{\sactivity, a\}, \{a, \eactivity\})$ cannot easily be pruned. Both behaviors of $L$ fit on $p'$.

We expect that many desirable final model properties can be encoded into similar evaluations to keep intermediate results invariant regarding that property. Another example are \emph{uniwired} Petri nets~\cite{Mannel2019FindingUniwired}. Two activities $a, a' \in \mathbb{A}$ are wired by a place $p \in \mathbb{P}$ if and only if $a \in {\bullet}p \wedge a' \in p{\bullet}$. Uniwired Petri nets have the property that for any two activities, they contain at most one place that wires them. This can easily be translated into a subtree-monotonic constraint. Specifically, if we accept a candidate $p$ into our intermediate solution, we can generate the constraint $\text{\textsc{not-biwiring}}_p(\hat{p}) :\Leftrightarrow ({\bullet}\hat{p} \cap {\bullet}p) = \emptyset \vee (\hat{p}{\bullet} \cap p{\bullet}) = \emptyset$. Note that this constraint is not as ``pretty'' (local) as the fitness-derived constraints, as it can exclude previously proposed candidates and is highly sensitive to ordering. However, it is very strong, i.e., it shrinks the candidate space dramatically, because it touches multiple disjoint subtrees.

\subsubsection{Concrete Composition State Update}
\label{sec:composition-state-update}

In the following, we briefly discuss what can be done in the \textsc{update-state} procedure of \Cref{alg:general-composing}. As mentioned earlier, the general version is very expressive with its arbitrary state but not yet fully utilized by existing instantiations. Currently, as we rely on existing eST-Miner variants, we are closer to the greedy simplification given in \Cref{alg:greedy-composing}. We want to note here that some variants~\cite{Mannel2019FindingUniwired,Mannel2022DiscoveringProcessModels} are actually what we would call \emph{relaxed greedy} as they have the ability to postpone this acceptance decision and collect postponed candidates in a heuristically-ordered priority queue. We refer the interested reader to the primary literature.

Greedy composition relies heavily on place-local evaluations. However, there is, of course, a gap between place-local evaluations and model-level properties. For example, while the place fitness evaluation $\mathit{eval}^\mathit{ind}_{\square,{\triangledownscript},\triangle}$ can successfully be used for filtering candidates and generating constraints, it is not sufficiently strong to guarantee fitness of the set of all accepted places in the case of $\tau < 1$. Naive application could easily produce even a deadlocked model, if only disjoint partitions of the log behavior are fitting on the constituent places. Refer to~\cite{Mannel2022DiscoveringProcessModels} by Mannel et al.\ and the therein introduced ``Delta Variant''.

Furthermore, evaluation relative to the intermediate result is useful for ensuring invariants or final properties of result sets. One such invariant is the absence of implicit places described before in \Cref{sec:concrete-evaluations}. Assuming an evaluation oracle, e.g., $\mathit{eval}^\mathit{rel}_\text{global implicitness} : \mathbb{P} \times \mathcal{P}(\mathbb{P}) \to \{\mathrm{True}, \mathrm{False}\}$, it is trivially incorporated into the framework. In our implementation, we provide an LP (linear program)-based and replay-based evaluator for $\mathit{eval}^\mathit{rel}_\text{global implicitness}$ and $\mathit{eval}^\mathit{rel}_\text{replay implicitness}$ respectively.

When this entire proposal, evaluation, and composition cycle terminates, either due to candidate exhaustion or otherwise signaled by time limit or surpassed quality threshold, we initiate one round of post-processing. Within this algorithmic framework, and particularly greedy composition, there are some areas where unwanted local structures can build up such that a ``cleanup'' step becomes integral to the overall discovery quality.

\subsection{Post-Processing}
\label{sec:post-processing}

We make post-processing a first class member of this discovery framework to provide an opportunity for fixing quirks that can come up in bottom-up discovery. Basically, allowing a last global \emph{look over} after being stuck in local decision-making for efficiency. In particular, we intend for this step to be pipelined with sequential execution of reusable transformations. From thorough structural implicit place removal after, e.g., a heuristic application during PEC-cycling, to place merging for increasing simplicity when the tree traversal depth is limited.

As we allow graceful premature cancellation of the PEC-cycling stage, it can also be used to fix artifacts that occur when not all possible places have been directly or indirectly considered. The aforementioned place merging is an example of that. Lastly, some technical aspects like model conversion can now be easily incorporated. Generally, any model repair strategy and particularly, existing implementation of it, could be applied here. This way, existing work can be easily reused and combined, which is one of the major goals of our framework implementation that we present in the next section.

\section{Implementation}
\label{sec:implementation}

In this section, we introduce our extensive implementation of the conceptual framework. The full development repository is available on GitHub\footnote{\url{https://github.com/leah-tgu/specpp}}. Additionally, we have published a bundled plugin with an interactive GUI frontend in the package \textsc{SPECpp} on the ProM nightly build\footnote{\url{https://www.promtools.org/doku.php?id=nightly}}. This frontend is also included in the GitHub repository.

A basic overview of the components is shown in \Cref{fig:pec-components}. The software components largely correspond to the elements we described in the previous section. On the proposal (left) side, we have the \textsc{Proposer} which uses the \textsc{Candidate Tree}. The tree is made up of its \textsc{Expansion Strategy}, and the crucial \textsc{Child Generation Logic} which is responsible for efficiently incorporating constraints. Proposed candidates are handed over to the \textsc{Composer} component that internally manages a \textsc{Composition} which stands in for the internal state that is updated using the \textsc{Individual} and \textsc{Relative Evaluators}. The composer may generate constraints for the proposer in response to a candidate.

\begin{figure}[htbp]
  \centering
  \includegraphics[width=0.8\linewidth]{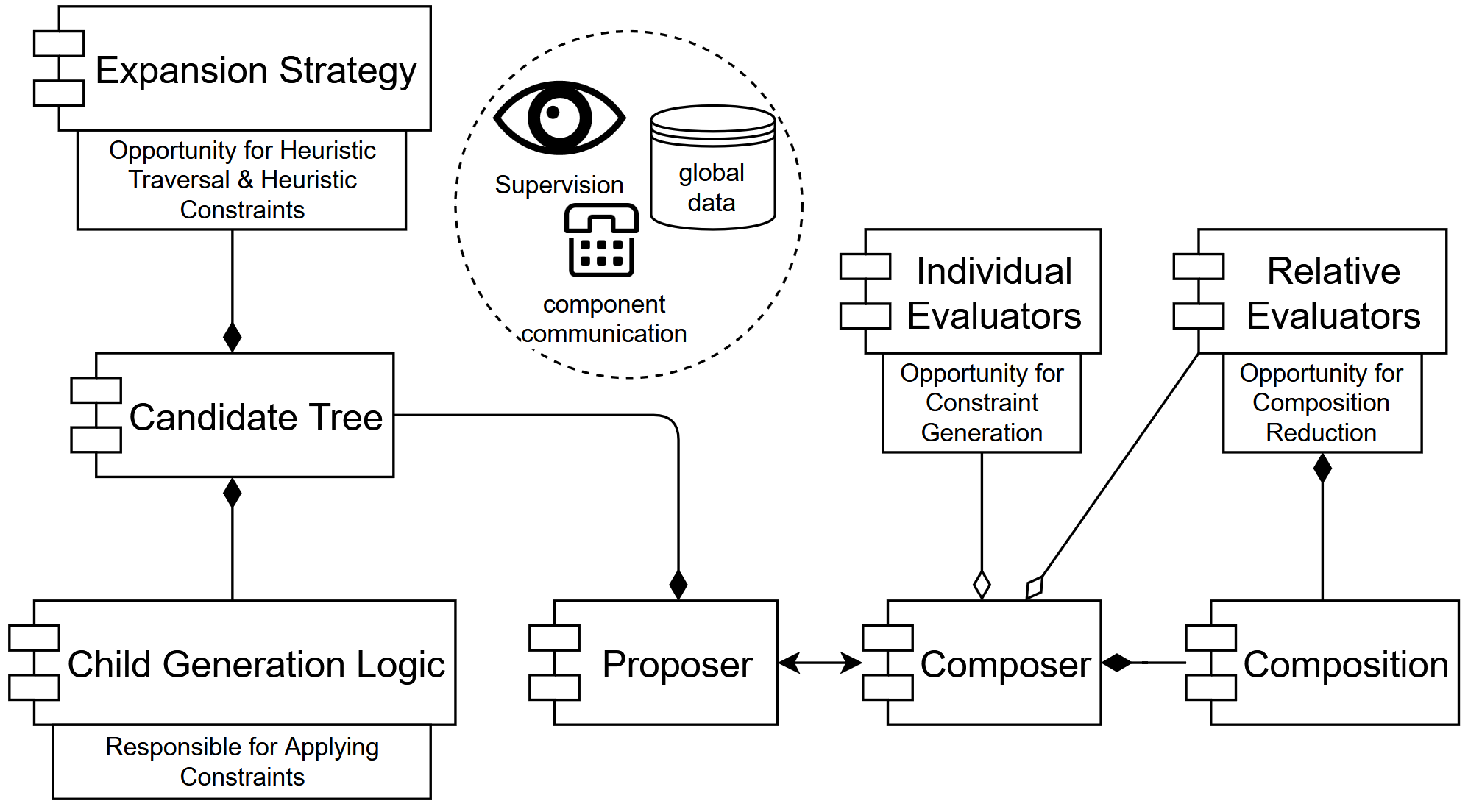}
  \caption{An overview of the main PEC components, their dependencies and their main responsibility. They mimic the conceptual structure with the addition of \emph{Supervision}.}
  \label{fig:pec-components}
\end{figure}

Additionally, there is the more framework implementation relevant supervision, shared global data storage and component communication linking service. Particularly the \textsc{Supervision} system is of note here. In this setting, where heuristic ordering, greedy decision making, constraints and unforeseen combinations of the former are abound, we consider it a first-class requirement to be able to inspect in some manner how a particular result came together. The internal execution, data and communication flow is so dynamic that in order to know \emph{what happened} during execution, it is of great benefit to allow developers to define custom descriptive events in their algorithm variants. These are asynchronously collected and handled flexibly in the supervision system. Arbitrary supervisors can be implemented to inspect the behavior even at runtime. For example, this system is used for a live discovery view in the interactive ProM plugin. More basic usage would be logging to a log file and/or the console. Particularly important for evaluating new proposal, evaluation or composition strategies is also a view on \emph{performance}. It is detailed to individual executions for future uses, however, even the aggregated average per task timing and occurrence frequency are massively useful in finding bottlenecks.

Supervision, as well as the other mentioned framework facilities, are realized by a requirement management system that permits dynamic at-runtime provisioning and requesting of data and functions across components. As we developed this software framework with the intention to support a high degree of flexibility, interfaces are very minimal and general. To make up for this, we provide the aforementioned systems to ease coordination and monitoring. Furthermore, the software framework facilitates quick prototyping through reuse via its focus on interoperability of different variants/implementations. Ideally, new implementations are split into separable components such that parts of these approaches can be reused and combined.

\subsection{Available Component Instantiations}
\label{sec:component-instantiations}

We fitted our framework implementation with some base components for the instantiations we described in \Cref{sec:conceptual-framework}. It explicitly supports heuristic expansion strategies, including the scoring metrics introduced in~\cite{Mannel2020ImprovingState}. Basic depth-first, which is most memory efficient, and breadth-first expansion are available as well. The aforementioned paper also introduces activity ordering strategies for computing $<_\mathrm{pre}$ and $<_\mathrm{post}$ using metrics derived from the input log. We include these as well as lexicographic and random orderings. Our child generation logic implementation supports the crucial preset/postset expansion cutoff constraints, depth limiting, wiring constraints and transition blacklisting. In contrast to our conceptual framework, we can permit non subtree-monotonic constraints at the cost of incomplete traversal. This occurs for example when we allow wiring constraints to be deleted after the wiring place is retroactively removed from the intermediate result set. Another avenue that is technically possible, is restarting the candidate tree traversal with the current constraints to get closer to completeness in these cases.

On the side of composers, we included the pattern of recursion/nesting. After applying their own decision making logic, recursive composers delegate the candidate handling to a child component. This allows reuse of ``partial'' composition strategies like fitness filtering. The main implemented strategies are the aforementioned fitness filtering corresponding to the base eST-Miner (refer to~\cite{Mannel2019FindingComplex}), a uniwired composer as introduced in~\cite{Mannel2019FindingUniwired}, the delta composer from~\cite{Mannel2022DiscoveringProcessModels} and an implementation for concurrent implicit place removal from~\cite{Mannel2020RemovingImplicit} supporting both the replay-based version as well as the linear-programming based structural version. Similarly, this nesting is supported for compositions to allow an outer composition to internally manage other compositions.

For evaluators, there is a standard and a parallelized version of token-based replay fitness which supports computation on log-subsets. Additionally, there is a facility for computing place token counts over traces which is used in replay-based implicit place removal. Implicit place checkers (both kinds) are also available as relative evaluators.

There are two ways to access the framework. For one, there is the headless package (i.e., not depending on ProM) in the repository. It contains main classes for command line access to simple ``event log in -- Petri net out'' execution, experimentation execution with lots of output and logging, and a command line interface for batch execution of freely configurable parameter variations. The latter also contains documentation and a usage guide for file-based component configuration. The other way to access at least a part of the capabilities of the framework is the ProM plugin.

\subsection{The ProM Plugin}
\label{sec:prom-plugin}

The accompanying ProM plugin is layered on top of the framework implementation itself. It serves the purpose of providing a visual and interactive interface to advanced process mining users who merely want to access already implemented variants. Moreover, the ProM infrastructure enables combining this approach with more than 1500 other plugins. Due to the extensive configuration options and their importance, it may not be very accessible to people not familiar with the research behind it. Even though we do provide presets, some deeper understanding is required for a proper interpretation of the results. The plugin and the ProM platform in general are not really geared towards ``true process mining end users''. In the following, we go into detail about its usage.

After selecting an event log input and calling the \textsc{Interactive SPECpp} plugin, the user is presented with an interactive GUI that sequentially guides them through configuration and execution of the underlying \textsc{SPECpp} discovery. The procedure is divided into the stages \emph{Pre-Processing}, \emph{Configuration}, \emph{Discovery} and \emph{Results}. Dynamic indicator buttons on the top row indicate the current, as well as currently enabled, stages. The active stage is marked in pink and its label is underlined. Accessible stages are marked in blue, whereas inaccessible ones are greyed out. This is exemplified in \Cref{fig:prom-plugin}. The user is in the discovery stage and can at any time go back to pre-processing and configuration but cannot skip ahead to the results because the computation is still running. Selecting a previous (to the left of the current stage) stage, resets the internal progress to that stage. For the stages concerning configuration, the last executed configuration is restored. This enables fast experimentation with parameters and settings.

\paragraph{Pre-Processing}

This stage gives an opportunity to preview the loaded data after applying an event classifier and an activity ordering strategy. It is further possible to select only a subset of activities to consider for place presets or postsets respectively.

\paragraph{Configuration}

This view has a dynamic list of options categorized into each of the framework's major components, as well as a section for parameterization. The first section offers presets for the eST-Miner variants proposed by Mannel et al.~\cite{Mannel2019FindingComplex,Mannel2019FindingUniwired,Mannel2022DiscoveringProcessModels}. Additionally, supervision can be turned off (or reduced) for increased performance if no insight into the execution specifics is required. Generally, only a simplified, manually designed component implementation selection is possible via this interface, however, the post-processing remains as generic as the underlying framework itself. The reason for this abstraction is usability. The full feature set is accessible to developers via code-based and \texttt{.json} file-based configuration. We recommend interested researchers to preferably check out the GitHub page.

\paragraph{Discovery}

The discovery view presented in \Cref{fig:prom-plugin} is constantly asynchronously updated while the PEC-cycling is performed. The left panel displays the currently accepted places either as a list or graph. The user can select the update frequency. For a high number of places, graph layouting becomes useless (and slow), so it is disabled. The right hand side is separated into four sections which are also all continuously updated. From top to bottom, there are sub panels for the search space, execution progress, performance, and monitored events.

\begin{figure}[htbp]
  \centering
  \rotatebox{90}{%
    \begin{minipage}{\textheight}
      \centering
      \includegraphics[
        width=\linewidth,
        height=\textwidth,
        keepaspectratio
      ]{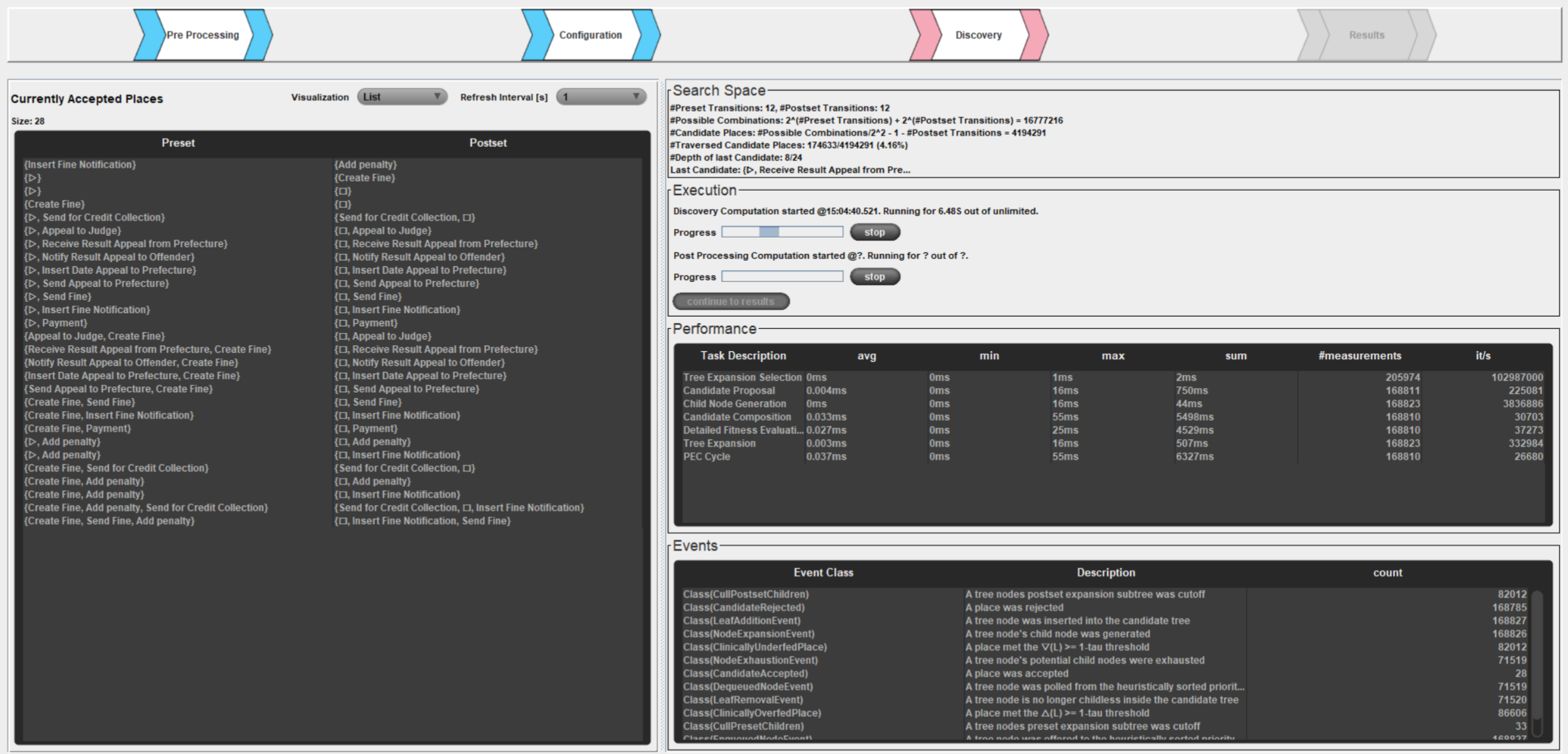}

      \captionof{figure}{The discovery view of the ProM plugin. The left split shows the currently accepted places in a list view. A graph visualization is available as well. The right side provides real-time information about the computation regarding candidate space traversal, performance, and events.}
      \label{fig:prom-plugin}
    \end{minipage}%
  }
\end{figure}

The search space panel gives an overview of the maximum number of candidate places, as well as the as-of-now computed fraction of those. The indicator of the last candidate's tree depth can additionally give a hint towards the remaining processing time. The elapsed time for the main PEC-cycling (place \emph{discovery}) and final post-processing computation is displayed on the execution progress panel. If time limits were configured, they are shown as well. There is additionally the option to manually cancel the computation at any time. Cancellation of the first stage is \emph{graceful}, i.e., the current intermediate result can be used for further post-processing. When both steps terminate successfully, the ``continue to result'' button is enabled.

The performance panel consists of a table that summarizes all currently observed (via the supervision system) performance events. The columns are selected summary statistics such as mean, min, max and frequency. Any developer-defined custom events are caught here as well. So, for example a newly implemented heuristic's performance can be instrumented as well. Similar to the performance table, the event table lists all currently observed descriptive events generated by the components and their frequencies. This helps with understanding how often certain paths are taken in proposal or composition decisions, e.g., how many underfed places are detected or how many implicit places were rejected. Any custom developer-defined events can be displayed here. If supervision was disabled or reduced, the corresponding tables remain empty.

\paragraph{Results}

The final resulting Petri net, as well as potentially visualizable intermediate post-processing results are displayed. Fitness for the entire model, as well as individual places in the initially resulting place set can be displayed. Additionally, overall alignment-based fitness and ETC precision are calculated. From the export panel, the user can export the final resulting Petri net and used algorithm configuration.

\section{Evaluation}
\label{sec:experiments}

To evaluate our approach, we performed automated experiments on synthetic as well as real-life event logs. The intent is to demonstrate the capabilities of the basic instantiations of conceptual components to achieve acceptable results on popular evaluation datasets. Furthermore, we highlight features of our implementation. We use the synthetic \textsc{Repair, Reviewing} and  \textsc{Teleclaims} event logs from~\cite{vanderAalst2016DataScience}, and real-life logs  \textsc{HospitalBilling}\footnote{\url{https://doi.org/10.1007/978-3-319-59536-8_34}},  \textsc{Sepsis}~\cite{Mannhardt2016Sepsis} and  \textsc{Road Traffic Fine Management (RTFM)}~\cite{deLeoni2015RoadTraffic}. As well as some past BPI challenge logs of 2012\footnote{\url{https://doi.org/10.4121/uuid:3926db30-f712-4394-aebc-75976070e91f}} and the  \textsc{Domestic Declarations (DomDecl)} log from the challenge's 2020\footnote{\url{https://doi.org/10.4121/uuid:52fb97d4-4588-43c9-9d04-3604d4613b51}} rendition. This selection of evaluation targets is diverse in its structure and represents staples in this discipline. \Cref{tab:event-logs} gives an overview. A high number of variants slows down fitness evaluation but most importantly, the number of occurring activities determines the number of possible candidate places. We can already see that BPIC12 presents an almost adversarial input.

We applied the batch execution feature of our implementation to perform discovery over a range of parameterizations.

\begin{table}[htbp]
  \centering
  \caption{The event logs used for our evaluation. The real-life event logs contain considerable variation. 80\% of traces in the \textsc{Sepsis} event log are unique. Whereas our artificial logs are quite small, the \textsc{RTFM} log contains over half a million events. Also, the BPI challenge logs sport a high number of activities.}
  \label{tab:event-logs}
  \begin{tabular}{lrrrr}
\hline
\textbf{Event Log} & \textbf{\#Traces} & \textbf{\#Variants} & \textbf{\#Activities} & \textbf{\#Events} \\
\hline
Teleclaims & 3,512 & 12 & 11 & 24,825 \\
Repair & 1,104 & 77 & 8 & 11,855 \\
Reviewing & 100 & 96 & 14 & 2,278 \\
HospitalBilling & 100,000 & 1,020 & 18 & 451,359 \\
Sepsis & 1,050 & 846 & 16 & 15,214 \\
RTFM & 150,370 & 231 & 11 & 561,470 \\
BPIC12 & 13,087 & 4,366 & 24 & 262,200 \\
BPIC20.DomDecl & 10,500 & 99 & 17 & 56,437 \\
\hline
\end{tabular}
\end{table}

For the component configuration, we chose the \emph{AverageFirstOccurrenceIndex} activity ordering from~\cite{Mannel2020ImprovingState}, the therein defined eventually-follows interest score for heuristic tree traversal and $\tau$-thresholding composition with concurrent implicit place removal from~\cite{Mannel2020RemovingImplicit}. As post-processing, we selected \emph{LP-based implicit place removal} from~\cite{Mannel2019FindingComplex} and \emph{Self-loop place merging}. These transformations merely simplify the model and do not change its supported behavior. Due to the not fully characterized nature of replay-based implicitness checking from~\cite{Mannel2020RemovingImplicit} in the $\tau < 1$ regime (i.e., log and model not exhibiting the same language), we forego using it as a pre-filtering step before the thorough LP-based version. This would be very beneficial for performance but as we discovered, it does not consistently underestimate implicitness. This selection of components is one of the simplest and nicely showcases the use-case of plugging together methods that were originally proposed in separate works.

The parameters were varied according to \Cref{tab:parameters}. The actual files, code and results are available on the project's GitHub page\footnote{\url{https://github.com/leah-tgu/specpp/releases/tag/paper_eval}}. $\tau$ is the threshold for the fraction of fitting traces as an acceptance criterion ($\#^\square_{L,p}$), as well as constraint generation in case the place is under/overfed on more than $(1 - \tau) \cdot 100\%$ of traces (as introduced in \Cref{sec:concrete-constraint-generation}). The tree depth limit almost determines runtime by severely limiting the exponential activity combinations. As the replay-based implicit place removal comes with no formal guarantees for most of our parameterizations as mentioned above, it turns into a mere heuristic. So, we test both versions \emph{RegionBased} ($\mathit{eval}^\mathit{rel}_\text{replay\_implicitness}$) and \emph{LPBased} ($\mathit{eval}^\mathit{rel}_\text{global\_implicitness}$) as our relative implicitness oracle during composition. We want to emphasize that limiting the tree depth to these low values compared to the full depth of $2 \cdot \mathit{\#Activities}$ is a reasonable choice as places with more than, e.g., six arcs connected to them, are detrimental to visual simplicity. The place depth acts as a very natural low-level complexity measure. Additionally, restricting control-flow dependencies to such few arcs is still very powerful. Few manually designed models will connect more than six activities to a place.

\begin{table}[htbp]
  \centering
  \caption{The varied parameters and their ranges. They were varied independently over their entire ranges resulting in $6 \cdot 5 \cdot 2 = 60$ combinations. The exact configuration files and code are available on the project's GitHub page. Decreasing $\tau$ and increasing the max tree depth increases the number of possible fitting places in the search space.}
  \label{tab:parameters}
  \begin{tabular}{lr}
\hline
\textbf{Parameter} & \textbf{Range} \\
\hline
$\tau$ & $\{1.0, 0.9, 0.8, 0.7, 0.6, 0.5\}$ \\
max tree depth & $\{2, 3, 4, 5, 6\}$ \\
implicitness calculator & $\{\text{RegionBased, LPBased}\}$ \\
\hline
\end{tabular}
\end{table}

We limited the execution time to 10 minutes for PEC-cycling, after which we gracefully stopped the loop and initiated post-processing which was also limited to 10 minutes before the computation was canceled. The LP-based post-processing in its current naive implementation has to solve quadratically many quadratically sized LPs (in the number of collected places), so if it does not terminate timely, it is a good indication that the resulting model would contain hundreds (or thousands in some cases) of places, making it thoroughly impractical to use and evaluate.

\Cref{tab:runtimes} gives some summary statistics on the recorded runtimes. Note that timeouts are not sharp because they rely on cooperative termination, thus the recorded execution times can exceed the set limits slightly. PEC timeouts are not that problematic as they are soft/graceful cancellations, i.e., the result is still being used. The total timeouts, i.e., those where a hard cancellation occurred, are more common among the bigger real-life logs. Particularly logs with high numbers of unique activities are problematic. \Cref{tab:collected-places} illustrates that for executions that did not finish (DNF), the number of collected places in the intermediate result, i.e., those meeting the $\tau$ threshold which were not identified as implicit to the rest, is indeed immense. It is not a loss to not consider such models. Rather, this points out the need for more sophisticated filters and candidate acceptance rules which are already being investigated by Mannel et al. A recent result is~\cite{Gross2023EnhancingApplicability}.

\begin{sidewaystable}[htbp]
  \centering
  \caption{Summary of the algorithm runtimes over all 60 attempted parameterizations. \textsc{BPIC20} executions often ran into the hard-cancellation timeout due to the high number of fitting places.}
  \label{tab:runtimes}
\resizebox{\textheight}{!}{%
\begin{tabular}{l@{\hspace{4.5pt}}c@{\hspace{4.5pt}}c@{\hspace{4.5pt}}c@{\hspace{4.5pt}}c@{\hspace{4.5pt}}c@{\hspace{4.5pt}}c@{\hspace{4.5pt}}c@{\hspace{4.5pt}}c@{\hspace{4.5pt}}c@{\hspace{4.5pt}}c@{\hspace{4.5pt}}c}
\hline
\textbf{Event Log} 
  & \multicolumn{2}{c}{\textbf{Timeouts}} 
  & \multicolumn{4}{c}{\textbf{PEC-cycling [s]}} 
  & \multicolumn{4}{c}{\textbf{Post-Processing [s]}} 
  & \textbf{Total [s]} \\
\cmidrule(lr){2-3} \cmidrule(lr){4-7} \cmidrule(lr){8-11} \cmidrule(lr){12-12}
& \textbf{PEC} & \textbf{Total} 
  & \textbf{10\%-pcrtl} & \textbf{avg} & \textbf{90\%-pcrtl} & \textbf{std} 
  & \textbf{10\%-pcrtl} & \textbf{avg} & \textbf{90\%-pcrtl} & \textbf{std} 
  & \textbf{avg} \\
\hline
Teleclaims & 0.0\% & 3.3\% & 0.02 & 14.30 & 28.98 & 48.71 & 0.01 & 28.52 & 4.68 & 118.13 & 42.82 \\
Repair & 0.0\% & 0.0\% & 0.02 & 1.13 & 2.35 & 2.56 & 0.00 & 1.28 & 0.55 & 8.08 & 2.41 \\
Reviewing & 0.0\% & 1.7\% & 0.14 & 9.31 & 28.16 & 16.67 & 0.00 & 10.68 & 1.58 & 77.40 & 19.99 \\
HospitalBilling & 33.3\% & 25.0\% & 1.79 & 293.34 & 600.04 & 269.42 & 0.15 & 181.76 & 600.00 & 255.08 & 475.09 \\
Sepsis & 11.7\% & 18.3\% & 1.64 & 174.38 & 600.01 & 215.16 & 0.08 & 129.61 & 600.00 & 228.42 & 303.99 \\
RTFM & 0.0\% & 15.0\% & 0.19 & 15.89 & 55.73 & 31.48 & 0.05 & 103.82 & 600.00 & 215.49 & 119.71 \\
BPIC12 & 56.7\% & 18.3\% & 17.85 & 407.49 & 600.71 & 244.93 & 0.16 & 202.00 & 600.00 & 245.41 & 609.49 \\
BPIC20.DomDecl & 31.7\% & 25.0\% & 0.81 & 253.26 & 600.15 & 265.78 & 0.94 & 179.29 & 600.00 & 219.73 & 432.55 \\
\hline
\end{tabular}%
}
\end{sidewaystable}

\begin{table}[htbp]
  \centering
  \caption{Overview of the median number of collected places in the intermediate and final results. The intermediate result size for canceled runs (DNF) where post-processing timed out reveal that in these cases, basic $\tau$-filtering and implicitness checking are simply not sufficient to handle the noise and complexity.}
  \label{tab:collected-places}
\begin{tabular}{lrrr}
\hline
\textbf{Event Log} & \makecell{\textbf{\#Collected} \\ \textbf{Places} \\ \textbf{(median)}} & \makecell{\textbf{\#Collected} \\ \textbf{Places (DNF)} \\ \textbf{(median)}} & \makecell{\textbf{\#Places after} \\ \textbf{Post-Processing} \\ \textbf{(median)}} \\
\hline
Teleclaims & 43 & 1247 & 14 \\
Repair & 22 & - & 10 \\
Reviewing & 36 & 675 & 13 \\
HospitalBilling & 111 & 4072 & 27 \\
Sepsis & 60 & 2437 & 16 \\
RTFM & 45 & 2710 & 14 \\
BPIC12 & 99 & 907 & 26 \\
BPIC20.DomDecl & 95 & 6826 & 24 \\
\hline
\end{tabular}
\end{table}

Further, we appended an automatic evaluation of \emph{alignment-based fitness} and \emph{ETC precision}~\cite{Buijs2012Role}. We also computed their harmonic mean, the \emph{F1-score}. As alignment computation on huge models with low fitness can take prohibitively long, we also limited the execution time to 10 minutes for that task. Similar to the post-processing above, an exceeded timeout is a good indication for an undesirable model---either due to its complexity or inadequate fitness.

Consider \Cref{fig:pareto-fitness-precision}, \Cref{fig:pareto-fitness-places} and \Cref{fig:pareto-precision-places} that show the trade-offs between the metrics fitness, precision, and number of places (stand-in for complexity, so lower is better) for the models with a top 25\% F1-score. We see that acceptable fitness and precision are achievable on any input log. However, we immediately notice that on the complex and noisy real-life logs, these scores come at the cost of large models. We rediscover the above finding that for these inputs, many feasible non-implicit places are found. This is also a consequence of the exponential increase in possible activity combinations with higher numbers of activities, which these logs tend to contain. For example, we should expect a log with 24 possible activities to require more places to constrain its model's behavior than for 8.

\begin{figure}[htbp]
  \centering
 \includegraphics[width=0.9\linewidth]{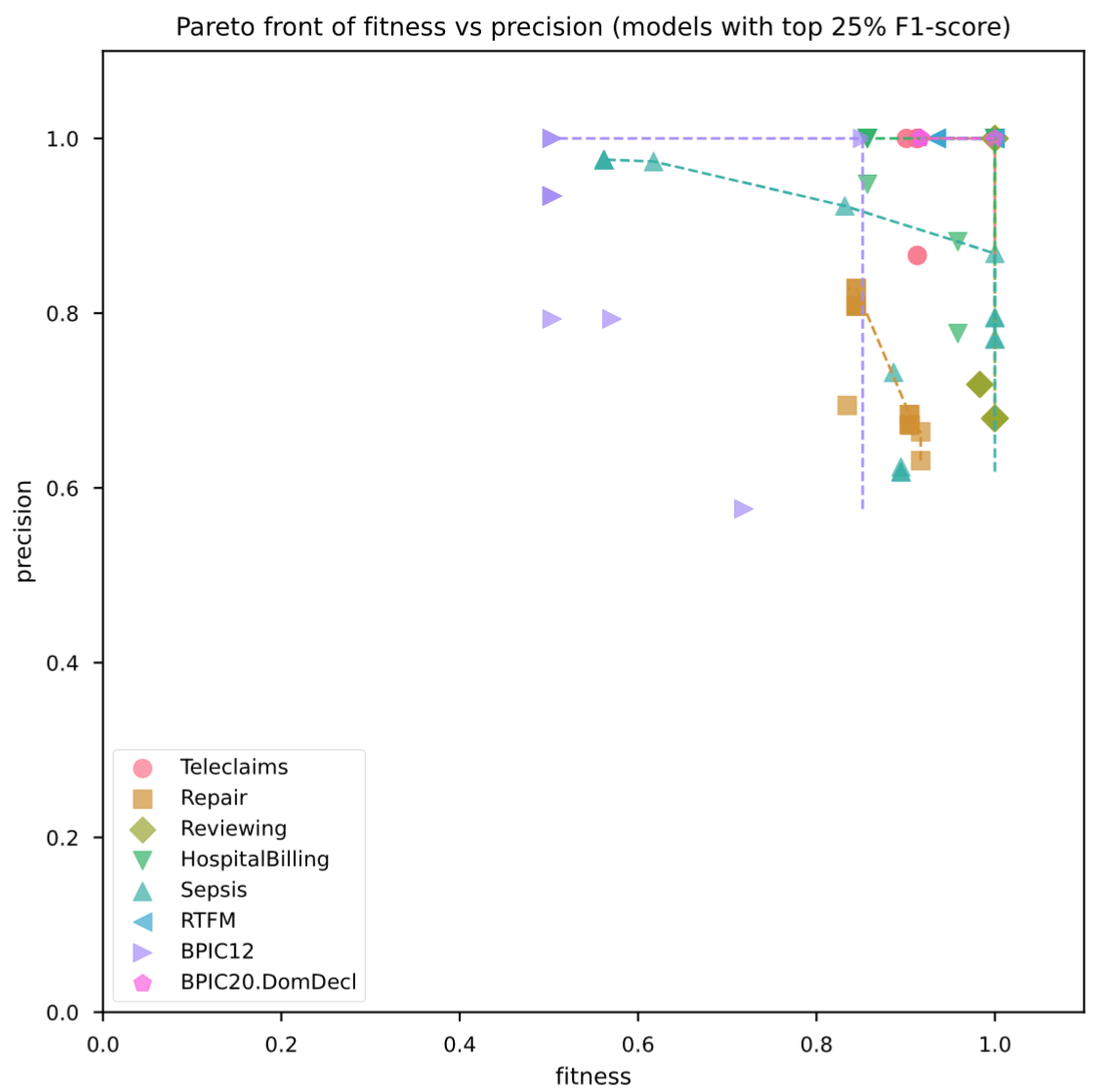}
  \caption{The distribution and Pareto front of fitness vs.\ precision of the best models (in terms of top 25\% F1-score).}
  \label{fig:pareto-fitness-precision}
\end{figure}

\begin{figure}[htbp]
  \centering
   \includegraphics[width=0.9\linewidth]{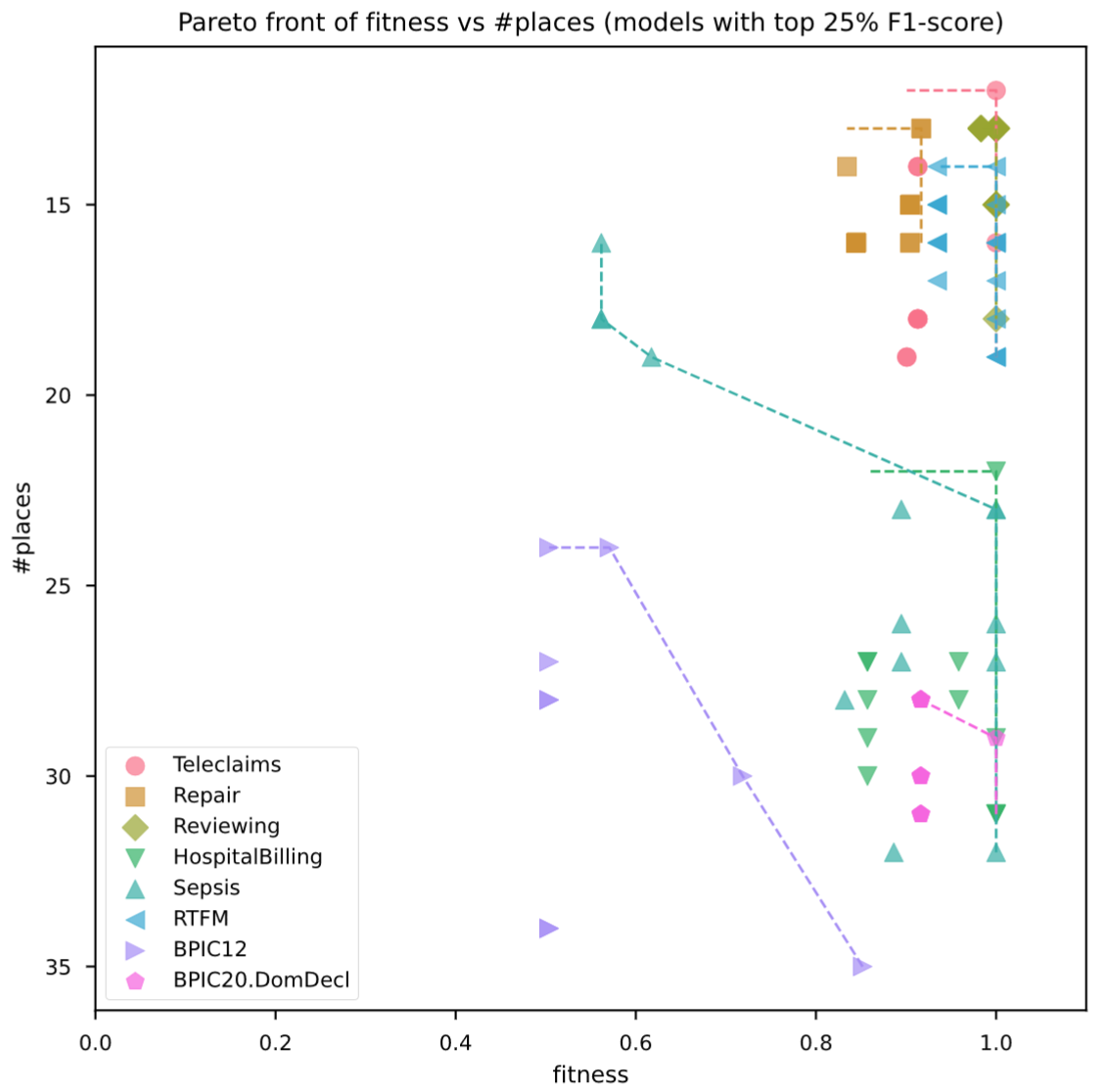}
  \caption{The distribution and Pareto front of fitness vs.\ \#places (lower is better) of the best models (in terms of top 25\% F1-score).}
  \label{fig:pareto-fitness-places}
\end{figure}

\begin{figure}[htbp]
  \centering
  \includegraphics[width=0.9\linewidth]{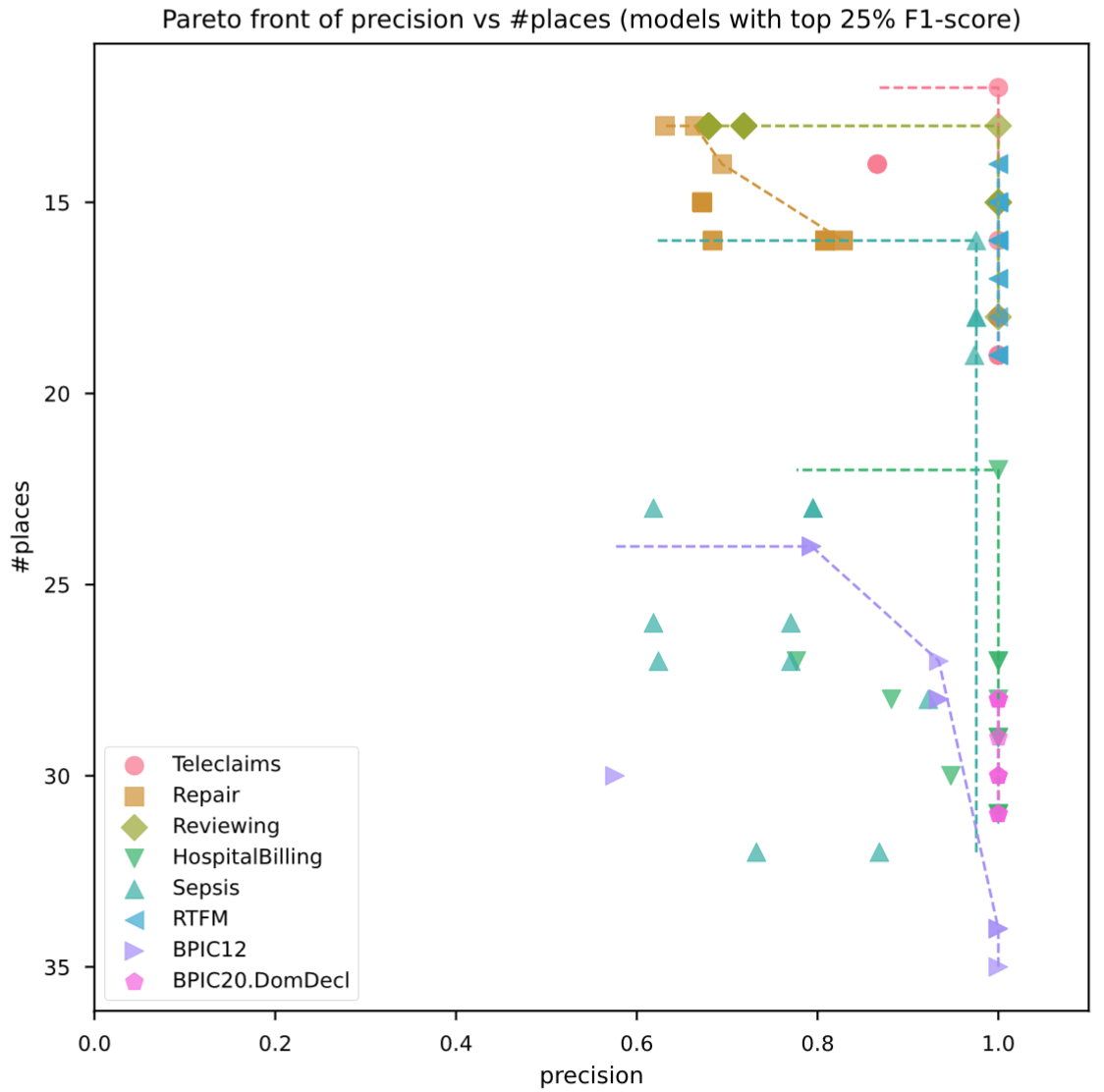}
  \caption{The distribution and Pareto front of precision vs.\ \#places (lower is better) of the best models (in terms of top 25\% F1-score).}
  \label{fig:pareto-precision-places}
\end{figure}

We want to emphasize that, naturally, a decrease in $\tau$ and increase in max tree depth lead to more feasible places and in turn longer running times. It is thus nice to see in \Cref{fig:percentile-rank} that the best scoring models can at least also be discovered in the faster runs. As we can see in \Cref{fig:minimal-tree-depth}, maximal fitness and also precision are usually reached even on low tree depths (and thus fast executions), which is in line with the previous observation. It also makes sense that, for higher precision, more constraining places are needed, while high fitness can be achieved with sufficient places to form a connected model.

For example, on \textsc{Teleclaims}, the best results with an F1-score of 1 were achieved in, e.g., 0.5s (depth limit 3), 8.3s (depth limit 4) or 48.67s (depth limit 6). This phenomenon is mainly caused by the implicitness evaluation, as at some point, almost every new candidate is implicit to one in the intermediate result. This highlights the importance of incorporating model property invariants like implicitness-freedom into constraint generation. We point to~\cite{vanderAalst2018DiscoveringGlue} and the notion of redundancy as another possible monotonic pruning property. Making this ordering agnostic for formal guarantees is not trivial, however. A practical remedy would be the detection of intermediate model quality improvement stagnation in a future composition strategy. This is manually achievable in our ProM live discovery interface previously shown in \Cref{fig:prom-plugin}. The graph or list visualization of the as-of-now accepted places makes it easy to see when they are not changing anymore. For example, when running the \textsc{RTFM} discovery interactively for $\tau = 0.7$, we can notice that the result stops changing after about three minutes, even if the full traversal ($\approx$ 6mil evaluations) takes 47 minutes.

Next, we turn our attention towards the effect of the $\tau$ and max tree depth parameterizations on the resulting model quality. \Cref{fig:correlation} shows the Spearman rank correlation coefficient (it assesses how monotonic the relationship between two variables is) between these parameters and fitness/precision. To look at their effect independently, we computed these for each of these two separately while controlling for the other, i.e., computing the correlation per controlled group and averaging them. There exists an interesting inversion in the relationship between $\tau$ and fitness for the smaller synthetic logs and the larger more noisy logs. While an increased $\tau$ has a positive impact on the former, it has a negative on the latter. That means for the more simplistic inputs, a higher $\tau$ leads to higher fitness, while complex logs require more less-than-perfectly locally fitting places to achieve an overall fitting model. Precision always requires lowered $\tau$ values. Probably because there are too few high-fitness places to properly constrain the resulting model otherwise. For the maximum tree depth, an increase seems to always have a positive influence on fitness and precision, though to varying degrees. On the synthetic logs, it is mostly important for precision, while on the real-life logs, fitness really benefits from more candidate space exploration.

\begin{figure}[htbp]
  \centering
  \includegraphics[width=0.9\linewidth]{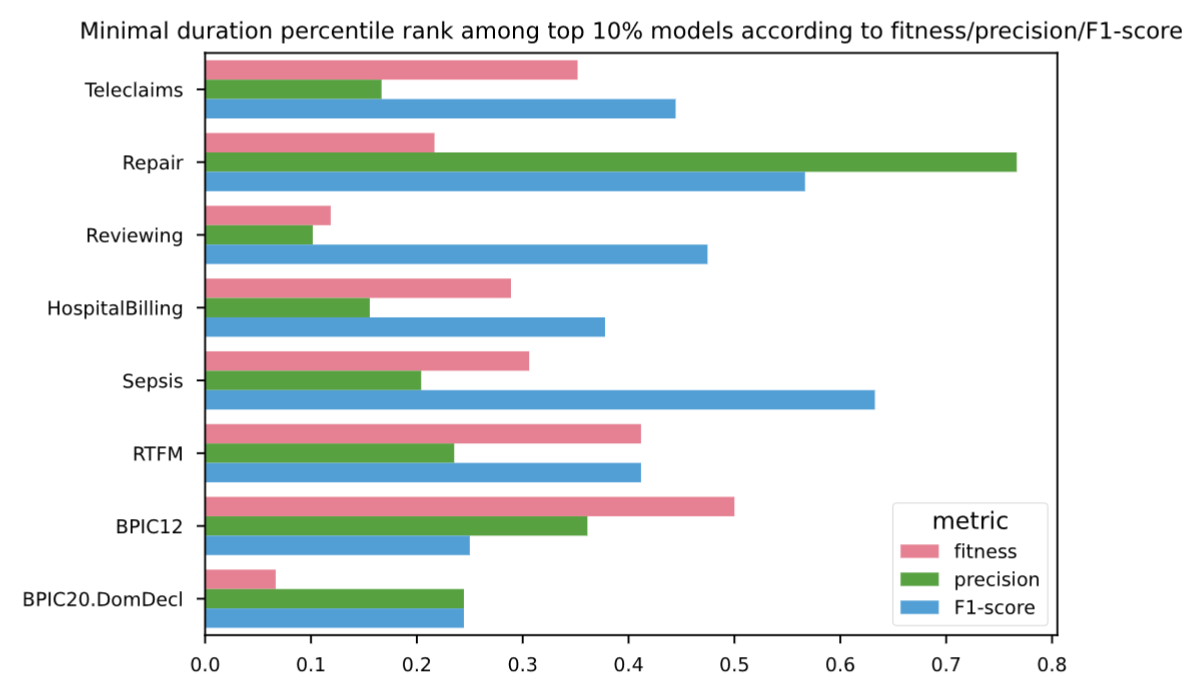}
  \caption{The percentile rank of the fastest execution to reach a model scoring in the top 10\% according to fitness/precision/F1-score. These fractions express that even the best models the algorithm can find for any parameterization are (also) discovered on the faster runs. It is clear to see that both in terms of fitness and precision, the best achievable results are consistently also achievable in the shortest possible time. As one might expect, to maximize them jointly (F1-score), more time is necessary. For example, on BPI20. \textsc{DomDecl} maximal fitness was reached by a run in the 6.7\%-percentile of runtimes, while to achieve maximal F1-score, a run in the 24.4\%-percentile was necessary.}
  \label{fig:percentile-rank}
\end{figure}

\begin{figure}[htbp]
  \centering
    \includegraphics[width=0.9\linewidth]{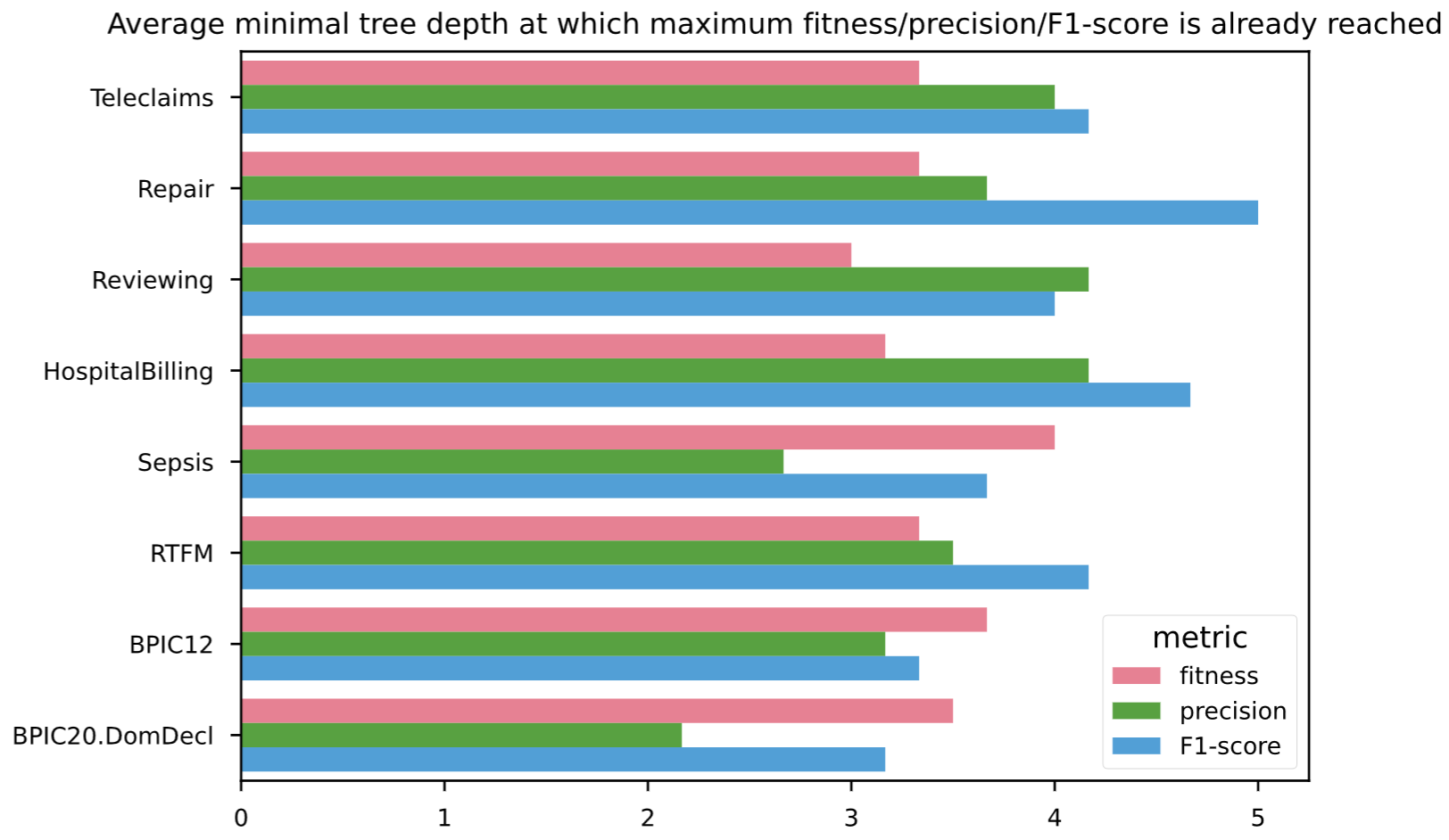}
  \caption{The average minimal \emph{max tree depth} at which the model quality metrics are already maximized at. The average is over the different values of $\tau$, as it has a big impact on the achievable metric. We see that, especially on the smaller synthetic logs, a high model fitness can be achieved with very low tree depth, thus saving running time. Precision appears to require more tree levels on average on these logs, and particularly to maximize both at the same time (F1-score), more candidate places are required.}
  \label{fig:minimal-tree-depth}
\end{figure}

\begin{figure}[htbp]
  \centering

  \begin{subfigure}{0.9\linewidth}
    \centering
    \includegraphics[width=\linewidth]{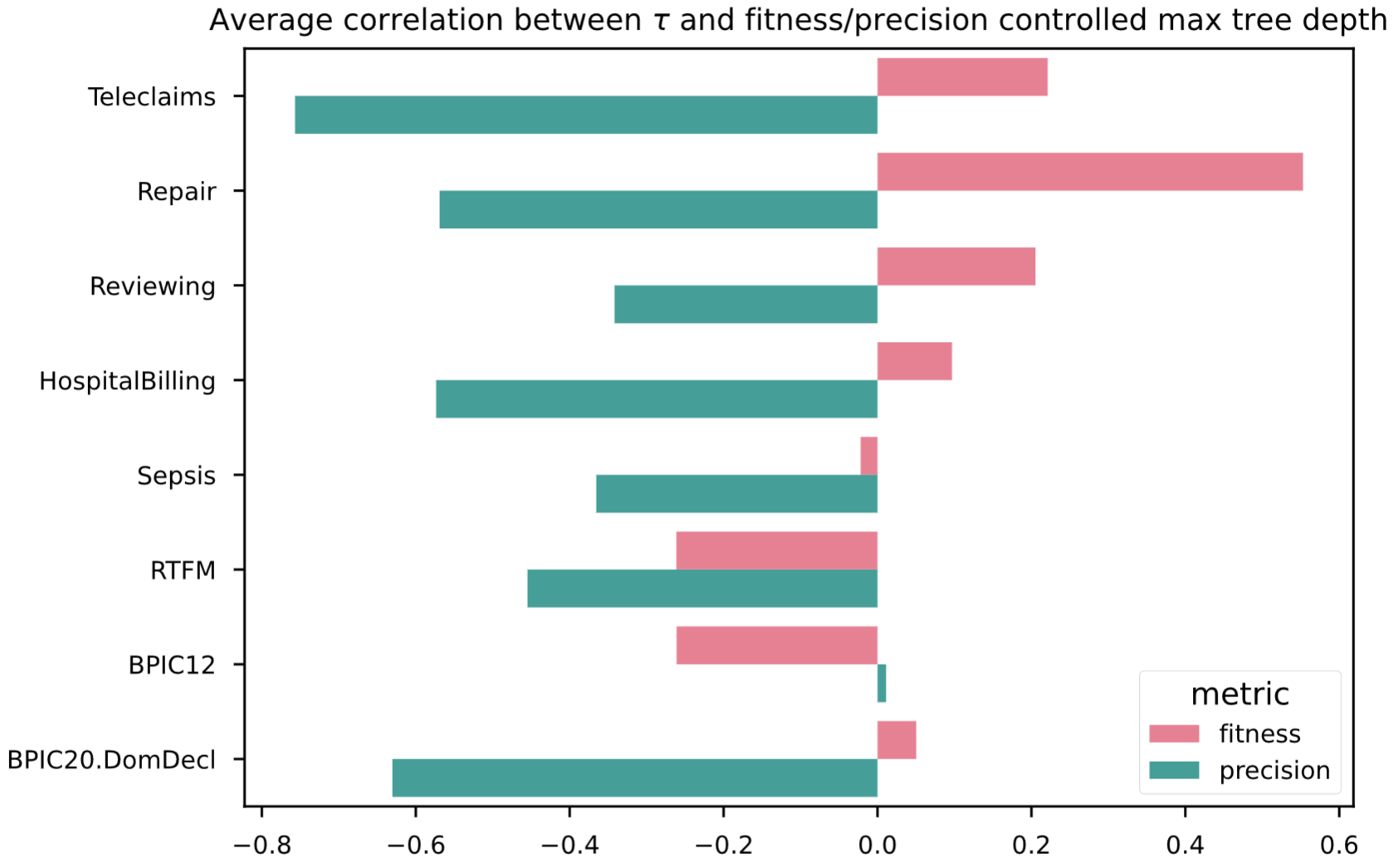}
    \caption{The average spearman correlation efficient between $\tau$ and the fitness/precision of the resulting models controlled by \emph{max tree depth}.}
    \label{fig:correlation-a}
  \end{subfigure}

  \vspace{1em}

  \begin{subfigure}{0.9\linewidth}
    \centering
    \includegraphics[width=\linewidth]{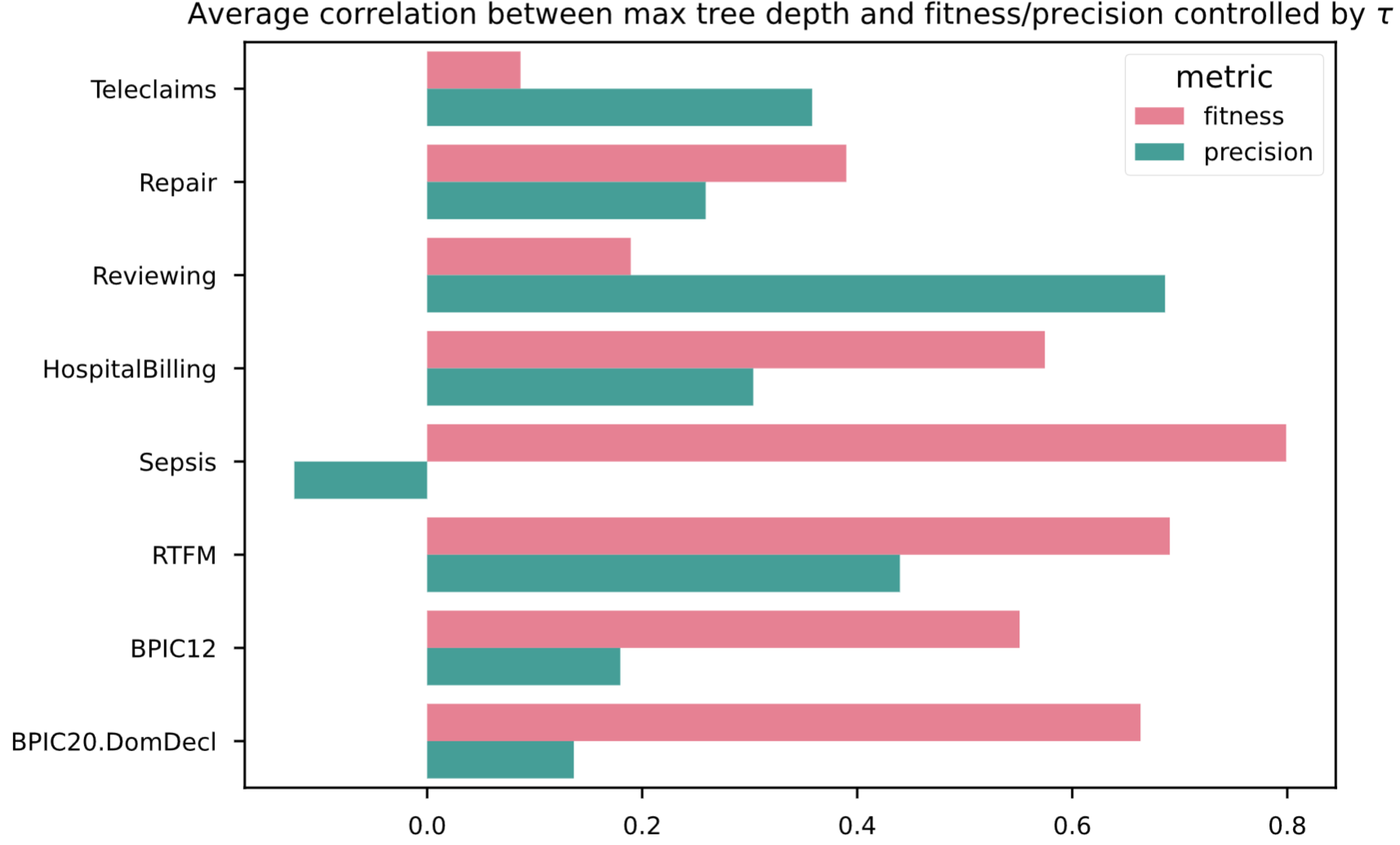}
    \caption{The average Spearman correlation efficient between \emph{max tree depth} and the fitness/precision of the resulting models controlled by $\tau$.}
    \label{fig:correlation-b}
  \end{subfigure}

  \caption{The correlation between the varied parameters and model metrics.}
  \label{fig:correlation}
\end{figure}

Lastly, we take a look at some of the concrete models we were able to discover during this evaluation. On the synthetic log \textsc{Teleclaims}, the model shown in \Cref{fig:model-teleclaims} with a fitness of 0.91 and precision of 1 shows on the one hand the tendency to discover more complex models but also the ability to provide high precision. Long-term dependencies like \emph{determine likelihood of claim} and \emph{initiate payment}, which are usually not discovered by other popular approaches, do complicate the model but may also provide some new insights. As the model was discovered for $\tau = 0.8$, we have the guarantee that each constraint represented by a place is applicable to at least 80\% of the observed behavior.

On \textsc{RTFM} which is a log with varied behavior, three executions with $\tau = 0.7$, max tree depth $\in \{3, 4, 5\}$ and LPBased implicit place removal produced an interesting model that is shown in \Cref{fig:model-rtfm}. It has an alignment-based fitness of 0.93, precision of 1 and perfectly fits 68\% of the log. We also see some of the quirks of the used basic $\tau$ threshold, as the activity \emph{Appeal to Judge} can never occur in a model trace due to its connected self-loop place. As that behavior apparently occurs in less than 70\% of traces, the place passes the filter. This strictly-said dead part could rather be regarded as an approximation.

\begin{figure}[htbp]
  \centering
  \rotatebox{90}{%
    \begin{minipage}{\textheight}
      \centering
      \includegraphics[
        width=\linewidth,
        height=\textwidth,
        keepaspectratio
      ]{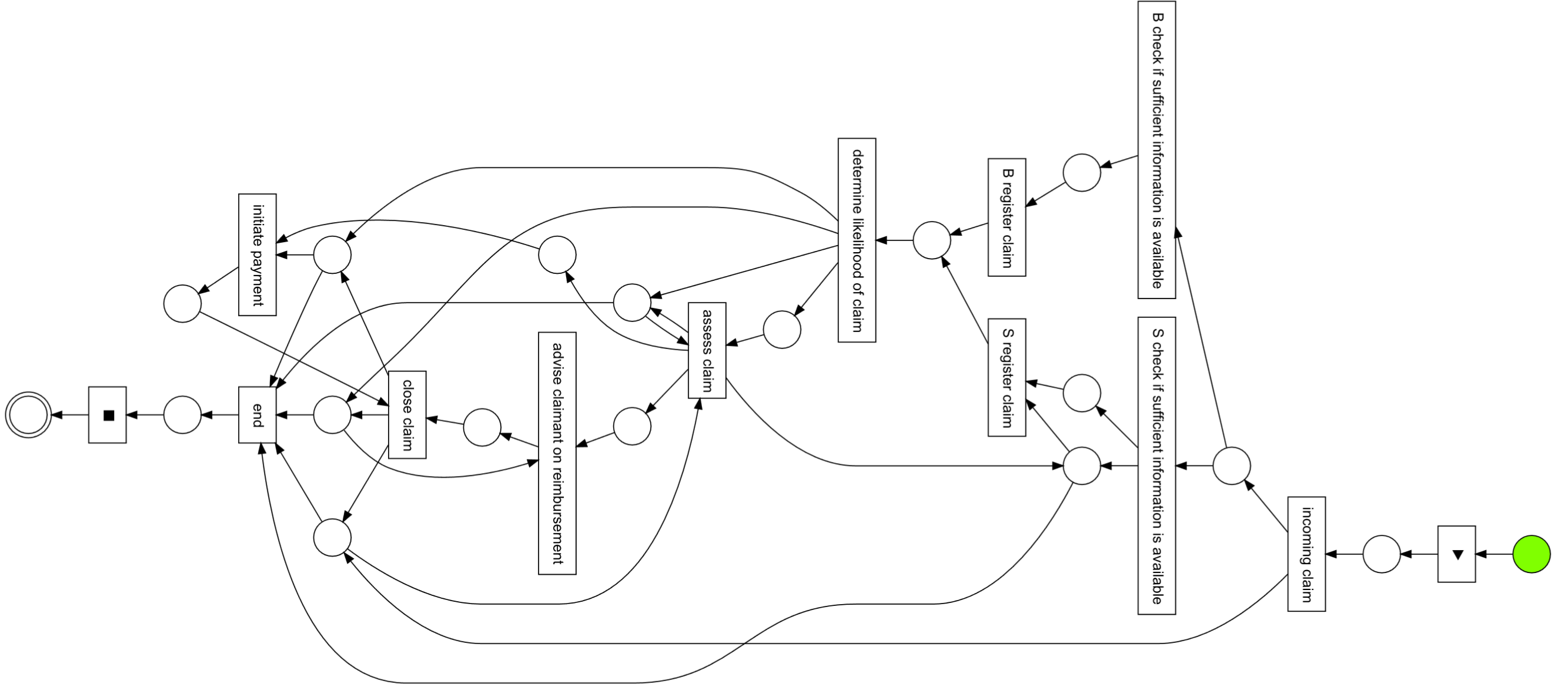}

      \captionof{figure}{One of the best models for \textsc{Teleclaims}, discovered for $\tau = 0.8$ and depths 4 and 5. It contains many long-term dependencies, e.g., between \emph{determine likelihood of claim} and \emph{initiate payment}. Its scores are 0.91 for fitness and 1 for precision.}
      \label{fig:model-teleclaims}
    \end{minipage}%
  }
\end{figure}

\begin{figure}[htbp]
  \centering
  \rotatebox{90}{%
    \begin{minipage}{\textheight}
      \centering
      \includegraphics[
        width=\linewidth,
        height=\textwidth,
        keepaspectratio
      ]{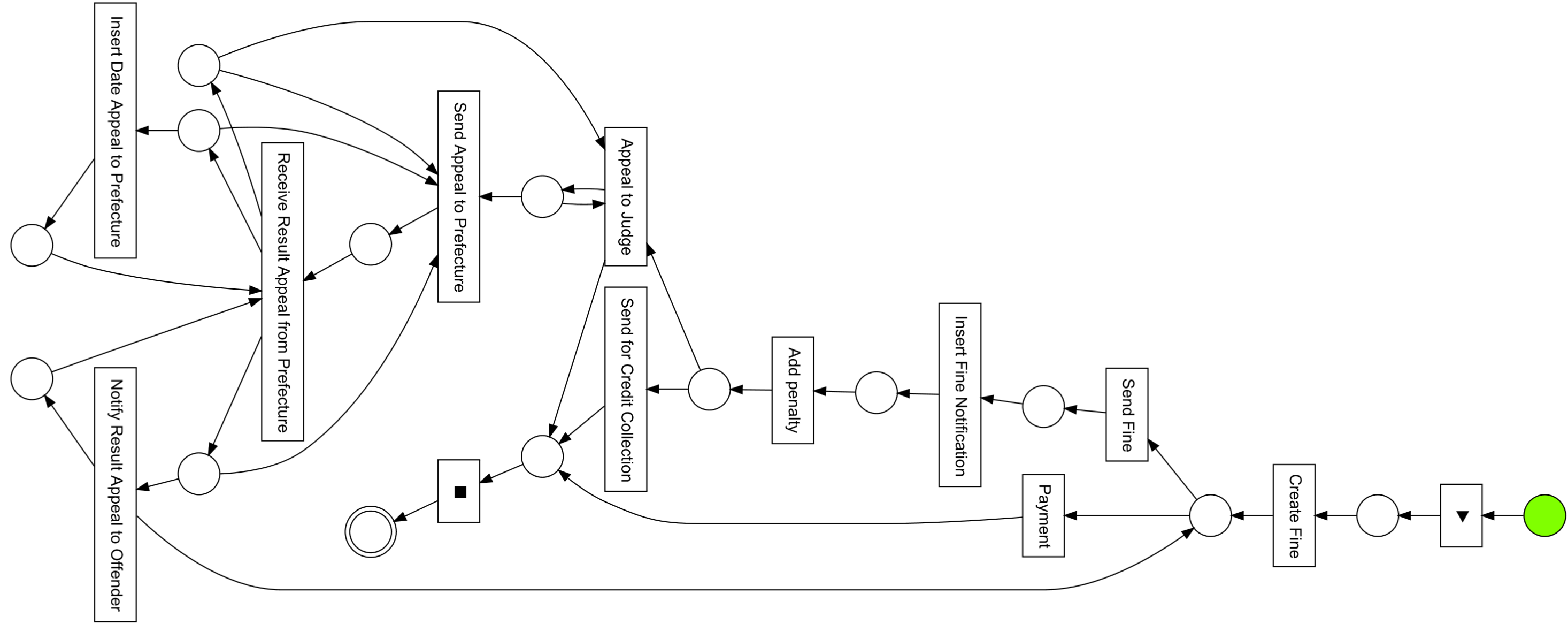}

      \captionof{figure}{One of the best models for the \textsc{RTFM} log. It was discovered on three different runs with $\tau = 0.7$ in 1.8s (depth limit 3), 10.4s (depth limit 4) and 45.3s (depth limit 5). It has a fitness of 0.93, precision of 1 and perfectly fits 68\% of the log.}
      \label{fig:model-rtfm}
    \end{minipage}%
  }
\end{figure}

\section{Related Work}
\label{sec:related-work}

As process discovery has provided a long-standing challenge in process mining~\cite{vanderAalst2003WorkflowMining,Tiwari2008Review,vanderAalst2013DiscoveringPetriNets}, many different approaches have been proposed \cite{vanderAalst2016DataScience,Ch2-PMhandbook-SS22,Ch3-PMhandbook-SS22}. Historically, it is predated by the Petri net synthesis problem concerned with creating a Petri net that precisely covers a behavioral specification~\cite{Badouel2015PetriNetSynthesis}. From this line of research stems a vast family of region-based approaches on languages by Carmona et al.~\cite{Carmona2008RegionBased}, Werf et al.~\cite{vanderWerf2009ProcessDiscovery}, van der Aalst et al.~\cite{vanderAalst2010ProcessMining} and van Zelst et al.~\cite{vanZelst2015AvoidingOverfitting,vanZelst2017DiscoveringWorkflow}. They mostly use (mixed) Integer Linear Programming (ILP) on a specifically designed optimization problem where solutions correspond to Petri net places. These approaches are very similar to ours in their aim to directly synthesize places that fit the observed behavior, thus constructing a Petri net in a bottom-up manner. They also avoid introducing restrictive modeling/representation bias and can therefore model non-free choice constructs. Historically, they struggle with noise and scalability issues. In contrast to our proposed framework, noise is typically tackled by pre-filtering as a consequence of the ``global'' view on place fitness. Whereas we can adaptively evaluate each candidate place locally and in dependence of the existing intermediate solution. There is also genetic process discovery~\cite{vanderAalst2005GeneticProcessMining,vanEck2014GeneticProcessMining,Tsai2010TimeInterval,deMedeiros2007GeneticProcessMining} which similarly has performance issues and cannot provide any guarantees about the resulting model.

Another early bottom-up work is the Alpha Miner by van der Aalst et al.\ in~\cite{vanderAalst2004WorkflowMining} which has seen many proposed variants~\cite{Li2007ExtendingAlpha,Wen2007MiningNonFreeChoice,Wen2010MiningPrimeInvisible,vanderAalst2022DiscoveringDirectlyFollows} which address noise handling and other limitations of the initial work. While the approach does synthesize individual places, and combines them into a full Petri net, the places themselves are derived based on extracted activity relations and not the full behavior of the log. Because of this, the approach effectively underlies stronger modeling bias as the patterns which are detectable in the activity relations are not perfectly expressive. Additionally, due to the simplicity of the algorithm---all places are discovered at once---places can also not be evaluated relative to each other, e.g., leading to deadlocked models.

The very successful and popular Inductive Miner by Leemans et al.~\cite{Leemans2013DiscoveringBlockStructured} and its variants~\cite{Leemans2013DiscoveringInfrequent,Leemans2016ScalableProcess,Brons2021StrikingNewBalance} can almost be viewed as a prototypical top-down approach. Working with the directly-follows graph of the given log, the idea is to recursively divide and conquer (solve) the discovery on sub graphs, and then collecting the sub problem solutions in the emerging tree structure. Strong representational bias is introduced in this step as the formalism used, process trees, only captures block-structured Petri nets. Additionally, as the patterns for high-level structure by which to split into sub problems are manually defined, modeling bias is introduced as well. Apart from the inability to discover non-free choice constructs or long-term dependencies, these approaches tend to underfit, providing low precision models. However, an important advantage is their support of silent (tau) transitions.

Building on the inductive miner framework, an algorithm that includes a bottom-up subroutine, precisely to address lacking precision, has also been proposed~\cite{Leemans2018IndulpetMiner}. The so-called Indulpet Miner is also process tree based, so it inherits the same representational bias as mentioned above. The bottom-up recursion tries to identify fitting partial cuts, essentially checking whether certain activity relations between the partial cut partitions are satisfied. Petri net places are more expressive than such predefined patterns but there are of course fewer cut candidates. Still, the authors mention the very high cost of full log traversal and the exponential complexity incurred in this step. As these points are precisely what our framework tries to address, it may be worthwhile to investigate an application of this concept to that technique, using partial cuts as candidates instead of places.

There are also discovery algorithms focusing on non-Petri net models. For one, there is the Split Miner proposed by Augusto et al.~\cite{Augusto2019SplitMiner}. It directly discovers simple BPMN models (using OR, XOR and AND-gateways) which are equivalent to block-structured Petri nets. They thus purposely introduce the same representation bias as the above mentioned inductive miner to keep the resulting models easily interpretable. However, their discovery of the BPMN gateways (splits and joins) is global and not fundamentally limited by the division into independent sub problems.

The heuristic miner by Weijters et al.~\cite{Weijters2011FlexibleHeuristics} discovers causal nets (C-nets) for example. The algorithm can deal with noise and provides slider-based abstraction to generate simple models. However, C-nets do not posses imperative semantics like the otherwise popular Petri nets and are thus not directly comparable.

Beyond our setting of case-centric event logs and control-flow models, bottom-up methods are also being investigated. For example, Agent System Mining~\cite{Tour2021AgentSystemMining}, a technique that composes multi agent systems from individual agent models, has been proposed by Tour et al. To lift the traditional case-centric process notion to a multi-agent one, they propose to mine for sub models and infer their interaction patterns. Similar to the object-centric paradigm~\cite{vanderAalst2019ObjectCentric}, a richer modeling formalism is required. One might consider such techniques to operate on a ``higher level'', as in, their candidates are \emph{models} instead of relatively atomic constructs such as Petri net places.

Lastly, as amply mentioned, 
the works of Mannel et al.~\cite{Mannel2019FindingComplex,Mannel2019FindingUniwired,Mannel2020RemovingImplicit,Mannel2020ImprovingState,Mannel2022DiscoveringProcessModels} on the so-called eST-Miner laid the groundwork for our approach. Our framework is a slight generalization of the former algorithms into this prototypical bottom-up setting.

\section{Conclusions}
\label{sec:conclusions}

In this paper, we presented \textsc{SPECpp}, a framework for bottom-up Petri-net discovery that casts discovery as the problem of efficiently identifying promising candidate places and composing them into a high-quality model. We formalized this process as a Proposal–Evaluation–Composition (PEC) cycle followed by post-processing.

A key advantage of bottom-up discovery is that it does not require the model to conform to a predefined block structure and can therefore discover complex constructs, including non-free-choice behavior and long-term dependencies.
At the same time, the flexibility of the framework and implementation does facilitate encoding restrictive design choices, e.g., uniwiring, in a seamless manner. The advantage directly leads into its biggest weakness: exponential complexity and thus running time.

To address the resulting exponential search space, we introduced a memory-efficient recursive representation of the candidate space that supports rigorous pruning using monotone constraints. This allows large sets of candidate places to be considered indirectly while only a fraction need to be evaluated explicitly.
Specifically, we used the locally (individual place) testable property, fitness, and aligned its monotonicity to our constraint system. Complementary to pruning are the candidate ordering prioritization mechanisms like heuristics.

On the side of candidate composition, we presented our general composition routine in which candidate places are evaluated and incorporated into an internal state, with the additional opportunity to generate constraints for future candidates. Further, we briefly described a greedy procedure that is based on relative evaluations to the intermediate result and candidate filtering based on existing eST-Miner procedures.

To accompany the conceptual framework, we provide a software framework for development as well as tool support with an interactive user interface inside the ProM ecosystem. A practical feature being the possibility to set a time limit, as the intermediate result can always be utilized.

In our evaluation, we demonstrated the applicability of the framework to synthetic as well as real-life event data. The evaluation demonstrates that the framework can discover high-quality models on both synthetic and real-life event logs, although its current runtime is not yet competitive with established discovery algorithms. Importantly, good models were often obtained long before exhaustive traversal of the candidate space, indicating that effective ordering and pruning strategies can make incomplete traversals practically useful.

We envision that this kind of discovery approach should be tuned to the user's needs. If there is enough time, a thorough traversal which actually provides some guarantees can be configured, otherwise a more heuristically guided instantiation can be used. The appeal lies in the sense of completeness that a well constrained candidate space traversal can provide.

Directions for future work emerge naturally from the structure of our conceptual framework. Starting with proposal, more complex heuristics that incorporate more log information and allow monotonic pruning of places would be a great start. Even slightly extending the $\tau$ threshold to relative fitness~\cite{vanderAalst2018DiscoveringGlue} may have potential. Further, an extended constraining strategy based on precision or other quality dimension metrics could improve candidate pruning speed in this first step while providing additional guarantees. Supporting discovery of silent transitions or non-uniquely labeled transitions would be a huge improvement.

Moving over to the less developed composition side, even more impactful improvements could be made. From defining a useful subclass of Petri nets that naturally provides subtree-monotonic constraints which would permit new guarantees on the final model (like uniwiring), to relaxing the current greedy strategy to combat running into local optima. The latter could be achieved by mirroring the tree structure of places for sets of places. That is, multiple compositions could be directly and indirectly considered. With monotone properties, scoring metrics, and at least approximate guiding heuristics, this extension would drastically improve the completeness of this approach. In general, this framework benefits from both rigorous provable candidate pruning and candidate composition space restrictions for model quality, as well as approximate heuristics used for ordering prioritization in incomplete traversals.

\paragraph{Acknowledgements}
We thank the Alexander von Humboldt (AvH) Stiftung for supporting our research. Additionally, we thank the Ministry of Culture and Science of the German State of North Rhine-Westphalia (MKW) and the Excellence Strategy of the Federal Government and the L\"{a}nder for supporting our research.

\bibliographystyle{spmpsci}
\bibliography{references}

\end{document}